\documentclass[12pt,reqno]{amsart}
\usepackage{amsmath,amsxtra,amssymb,amsthm,amsfonts,eufrak}
\numberwithin{equation}{section}
\usepackage[svgnames]{xcolor}
\usepackage{graphicx}
\usepackage{float}

\newcommand{\e}{\epsilon}

\newcommand{\kla}{\left ( }
\newcommand{\mer}{\right ) }
\renewcommand{\for}{\begin{eqnarray*}}
\newcommand{\mel}{\end{eqnarray*}}
\def\fr{\begin{align*}}

\newcommand{\kl}{\pl \le \pl}
\newcommand{\gl}{\pl \ge \pl}

\newcommand{\lel}{\pl = \pl}

\newcommand{\ten}{\otimes}

\newcommand{\p}{\hspace{.05cm}}
\newcommand{\pl}{\hspace{.1cm}}
\newcommand{\pll}{\hspace{.3cm}}

\newcommand{\ran}{\rangle}
\newcommand{\lan}{\langle}

\newcommand{\al}{\alpha}
\newcommand{\si}{\sigma}

\newcommand{\la}{\lambda}

\newcommand{\F}{{\mathcal F}}
\newcommand{\E}{{\mathcal E}}

\newcommand{\A}{{\mathcal A}}
\newcommand{\B}{{\mathbb{B}}}
\newcommand{\D}{{\mathcal D}}
\newcommand{\M}{{\mathcal M}}

\newcommand{\R}{{\mathcal R}}
\newcommand{\C}{{\mathcal C}}
\newcommand{\CC}{{\mathbb C}}

\newcommand{\N}{{\mathcal N}}
\newcommand{\T}{{\mathbb T}}

\newcommand{\fs}[2]{\frac{#1}{#2}}

\newcommand{\norm}[2]{\parallel \! #1 \! \parallel_{#2}}

\newtheorem{lemma}{Lemma}[section]
\newtheorem{prop}[lemma]{Proposition}
\newtheorem{theorem}[lemma]{Theorem}
\newtheorem{cor}[lemma]{Corollary}

\newtheorem{rem}[lemma]{Remark}

\newcommand{\re}{\begin{rem}\rm}
\newcommand{\mar}{\end{rem}}
\newtheorem{exam}[lemma]{Example}
\newcommand{\bra}[1]{\langle{#1}|}
\newcommand{\ket}[1]{|{#1}\rangle}
\newcommand{\ketbra}[1]{|{#1}\rangle\langle{#1}|}

\newcommand{\qd}{\end{proof}\vspace{0.5ex}}
\newcommand{\prf}{\begin{proof}[\bf Proof:]}

\newcommand{\xspace}{\hbox{\kern-2.5pt}}

\newcommand{\rra}{\rightarrow }

\allowdisplaybreaks
\newtheorem{defi}[lemma]{Definition}
\begin{document}
\title{Capacity Estimates via comparison with TRO channels}
\setlength{\parindent}{4ex}
\author{L. Gao}
\address{Department of Mathematics\\
University of Illinois, Urbana, IL 61801, USA} \email[Li Gao]{ligao3@illinois.edu}

\author[M. Junge]{M. Junge$^*$}\thanks{$^*$ Partially supported by $\text{NSF-DMS}$ 1501103}
\address{Department of Mathematics\\
University of Illinois, Urbana, IL 61801, USA} \email[Marius
Junge]{mjunge@illinois.edu}

\author[N. LaRacuente]{N. LaRacuente$^{\dag}$}\thanks{$^{\dag}$ This material is based upon work supported by NSF Graduate Research Fellowship Program DGE-1144245}
\address{Department of Physics\\
University of Illinois, Urbana, IL 61801, USA} \email[Nicholas LaRacuente]{laracue2@illinois.edu}
\maketitle

\renewcommand{\abstractname}{\bf{ABSTRACT}}
\begin{abstract}
A ternary ring of operators (TRO) in finite dimensions is a diagonal sum of spaces of rectangular matrices. TRO as operator space corresponds to
quantum channels that are diagonal sums of partial traces, which we call TRO channels. TRO channels admits simple, single-letter capacity formula. Using operator space and complex interpolation techniques, we give perturbative capacities estimates for a wider class of quantum channels by comparison to TRO channels. Our estimates applies mainly for quantum and private capacity and also strong converse rates. The examples includes random unitary from group representations which in general are non-degradable channels.
\end{abstract}

\maketitle
\section{Introduction}
Channel capacity, introduced by Shannon in his foundational paper \cite{48}, is the ultimate rate at which information can be reliably transmitted over a communication channel. During the last decades, Shannon's theory on noisy channels has been adapted to the framework of quantum physics. A quantum channel has various capacities depending on different communication tasks, such as \emph{quantum capacity} for transmitting qubits, and \emph{private capacity} for transmitting classical bits with physically ensured security. The coding theorems, which characterize these capacities by entropic expressions, were major successes in quantum information theory (see e.g. \cite{Wildebook}). For instance, the quantum capacity $Q(\N)$ of a channel $\N$, by Lloyd-Shor-Devetak Theorem \cite{Lloyd,Shor,Devetak}, is given by
\begin{align}\label{qcapacity}
 Q(\N)\lel\lim_{k\to \infty} \frac{Q^{(1)}(\N^{\ten k})}{k}\pl,\pl
Q^{(1)}(\N)\lel\max_{\rho }H(\N(\rho))-H(id\ten \N(\rho^{AA'}))\pl,
\end{align}
where $H(\rho)=-tr(\rho\log \rho)$ is the entropy function, and the maximum runs over all pure bipartite states $\rho^{AA'}$. Nevertheless, the capacities for many channels are computationally intractable due to \emph{regularization}, the limit in which one takes the entropic expression \eqref{qcapacity} over asymptotically many uses of the channel. Regularization is in general unavoidable, because the capacity of a combination of two quantum channels may exceeds the sum of their individual capacities \cite{SY,sa1,Cubitt}. This phenomenon, called ``super-additivity'', also exists for classical and private capacities \cite{Hastings,li2009,Elkouss}.

Devetak and Shor in \cite{DS} consider \emph{degradable channels}, for which the receiver can fully reproduce the information lost to the environment by ``degrading" the received output through another channel. Degradable channels are additive, admitting the trivial regularization $Q=Q^{(1)}$ and a simple ``single-letter'' formula for quantum capacity. Several different methods have been introduced to give upper bounds on particular or general channels (e.g. \cite{Holevo2, Winterss, sutter, sc, XW}). Little is known about the exact value of quantum capacity beyond degradable cases. In addition, it is desirable to know whether the strong converse theorem holds for quantum channels. The strong converse would mean that above the quantum capacity, there is a sharp trade off between the transmission rate and transmission accuracy. In this paper, we give capacities estimate for quantum channels via a new analysis of the Stinespring dilation. We briefly explain our main idea below.

Recall that a quantum channel $\N$ is a completely positive trace preserving (CPTP) map that sends densities (positive trace $1$ operators) from one Hilbert space $H_A$ to another $H_B$. $\N$ admits a \emph{Stinespring dilation} as follows
\begin{align}\N(\rho)=tr_E(V\rho V^*)\pl, \label{st}\end{align}
where $V:H_A \to H_B\ten H_E$ is a partial isometry and $H_E$ is the Hilbert space of the environment. We call the range $ran(V)\subset H_B\ten H_E$ the \emph{Stinespring space} of $\N$. Vice versa every subspace $X\subset H_B\ten H_E$ determines a quantum channel by viewing the inclusion as an isometry. Hence the capacities of a channel are determined by its Stinespring space, more precisely the operator space structure by regarding $ran(V)\subset H_B\ten H_E\cong \B(H_E,H_B)$ as operators from $H_E$ to $H_B$. This perspective was previously used in \cite{Guillaume} to understand Hastings' counterexamples for additivity of minimal output entropy.

A \emph{ternary ring of operators} (TRO) is a closed operator subspace $X$ closed under the triple product
\[x,y,z\in X \Rightarrow xy^*z\in X \pl .\]
TRO's were first introduced by Hestenes \cite{Hestenes}, and pursued by many others (see e.g. \cite{zettl,kaur}). In finite dimensions, TRO's are always diagonal sums of rectangular matrices ${\oplus_i (M_{n_i}\ten 1_{m_i})}$ (with multiplicities $m_i$), and the quantum channels whose Stingspring spaces are TRO's are diagonal sums of partial traces (Proposition \ref{pt1}). These simple channels have well-understood capacities \cite{wolf} and the strong converse property. Let $\N$ be a channel as \eqref{st} with its Stinespring space $ran(V)$ being a TRO in $\B(H_E,H_B)$. We consider the channel
\begin{align}\N_f(\rho)=tr_E\Big((1\ten f) V\rho V^*\Big)\pl, \label{mchannel}\end{align}
for which the Stinepsring dilation is modified by multiplying a operator $f$ on the environment $H_E$. With certain assumptions on $f$, $\N_f$ is also a quantum channel and we prove that the capacity of $\N_f$ is comparable to the original $\N$ in the following way,
\begin{align}\label{mainre} Q(\N)\le Q(\N_f)\le Q(\N)+\tau(f\log f)\pl
\end{align}
where $Q$ is the quantum capacity in \eqref{qcapacity} and $\tau(f\log f)=\frac{1}{|E|}tr_E(f\log f)$ is a normalized entropy of $f$. One class of our examples are random unitary channels arose from group representation. Let $G$ be a finite group and $u:G \to \B(H)$ be a (projective) unitary representation. For probability distributions $f$ on $G$, we define the random unitary
\[\pll \N_f(\rho)=\sum_g f(g)u(g)\rho u(g)^* \pl,\pl \N(\rho)=\frac{1}{|G|}\sum_g u(g)\rho u(g)^*\pl.\]
Here $\N$ is a special case of $\N_f$ with $f$ being the uniform distribution $(\frac{1}{|G|},\cdots,\frac{1}{|G|})$ on $G$, and its capacity $Q(\N)$ is given by the logarithm of the largest multiplicity in the irreducible decomposition of $u$. The inequality \eqref{mainre} implies that
\begin{align} Q(\N)\le Q(\N_f)\le Q(\N)+\log |G|-H(f)\pl,
\end{align}
where $|G|$ is the order of $G$ and $H(f)=-\sum f(g)\log f(g)$ is the Shannon entropy. When $G$ is a noncommutative group, $\N_f$ is in general not degradable especially when $f$ is close to the uniform distribution.

The key inequality in our argument is the following ``local comparison property'': for any positive operators $\si$ and $\rho$,
\begin{align}\label{pcom}\norm{\N(\si)^{-\frac{1}{2p'}}\N (\rho)\N(\si)^{-\frac{1}{2p'}}}{p}\pl\le \pl\norm{\N(\si)^{-\frac{1}{2p'}}&\N_f (\rho)\N(\si)^{-\frac{1}{2p'}}}{p}\pl \nonumber\\ &\le  \pl \norm{f}{p,\tau}\norm{\N(\si)^{-\frac{1}{2p'}}\N (\rho)\N(\si)^{-\frac{1}{2p'}}}{p}\pl,   \end{align}
where $\norm{a}{p}=tr(|a|^{p})^{\frac{1}{p}}$ is the Schatten $p$-norm, $\norm{f}{p,\tau}=\frac{1}{|E|}tr_E(|f|^{p})^{\frac{1}{p}}$ is the $p$-norm of normalized trace and $\frac{1}{p}+\frac{1}{p'}=1$. The ``local comparison property'' is actually an inequality of sandwiched R{\'e}nyi relative entropy introduced in \cite{MLDS,cconverse}. The sandwiched R{\'e}nyi relative entropies are used to prove the strong converse for entanglement-assisted communication \cite{ceaconverse}, and to give upper bounds on the strong converse of classical communication \cite{cconverse}, quantum communication \cite{qconverse}, and more recently private classical communication \cite{PLOB15,pconverse}. Based on these results, we find that our comparison method \eqref{mainre} also applies to strong converse rate for both quantum and private communication.

We organize this work as follows. Section 2 recalls the concept of TRO's from operator algebras and proves the ``local comparison theorem''. Section 3 is devoted to applications on estimating capacities, capacity regions and strong converse rates. Section 4 discusses examples from group representations. We provide an appendix describing the complex interpolation technique used in our argument. 
\section{TRO Channels and local comparison property}
\subsection{Channels and Stinespring spaces} We denote by $\B(H)$ the bounded operators on a Hilbert space $H$. We restrict ourselves to finite dimensional Hilbert spaces and write $|H|$ for the dimension of $H$. The standard $n$-dimensional Hilbert space is denoted by $\CC^n$ and $n\times n$ matrix space is $M_n$. A \emph{state} on $H$ is given by a density operator $\rho$ in $ \B(H)$, i.e.  $\rho\ge 0,\ {tr(\rho)=1}$, where ``tr'' is the matrix trace. The physical systems and their Hilbert spaces are indexed by capital letters as $A,B,\cdots$. We use superscripts to track multipartite state and their reduced densities, i.e. for a bipartite state $\rho^{AB}$ on $H_A\ten H_B$, $\rho^{A}=tr_B (\rho^{AB})$ presents its reduced density matrix on $A$. We use $1_{A}$ (resp. $1_n$) for the identity operator in $\B(H_A)$ (resp. $M_n$), and $id_{A}$ (resp. $id_n$) for the identity map on $\B(H_A)$ (resp. $M_n$).

Let $\N:\B({H_{A}})\to \B(H_B)$ be a quantum channel (CPTP map) with Stinespring dilation $\N(\rho)=tr_E(V\rho V^*)$. The complementary channel of $\N$ is
 \begin{align} \label{dilation}\N^E:\B(H_{A})\to \B(H_E)\pl, \pl  \N^E(\rho)=tr_B(V\rho V^*)\pl .\end{align}
This dilation \eqref{dilation} is not unique, but different ones are related by partial isometries between the environment systems. Given an orthonormal basis $\{\ket{e_i}\}$ of $H_E$ and its dual basis  $\{\bra{e_i}\}$ in $H_E^*$, one can identify the tensor product Hilbert space $H_B\ten H_E$ with the operators $\B(H_E,H_B)$ as follows,
\[\ket{h}=\sum_i \ket{h_i}\ten \ket{e_i} \to h=\sum_i \ket{h_i}\ten \bra{e_i}\pl, \pl\pl \ket{h_i}\in H_B.\]
This identification depends on the choice of the basis $\{\ket{e_i}\}$ but is unique up to a unitary equivalence. It acts as a partial trace on pure bipartite states,
\begin{align}tr_E(\ket{h}\bra{k})=hk^* \pl\pl, \pl\pl tr_B(\ket{h}\bra{k})=k^*h\pl. \label{ptrace}\end{align}
Throughout this paper we will use ``bra-ket'' notation for vectors and dual vectors. The Stinespring space $X=ran(V)$ then becomes an operator subspace of $\B(H_E,H_B)$.
Note that $X$ equipped with Hilbert-Schmidt norm is isomorphic to the input system $H_{A}$ via $V$. We can identify $\ket{h}$ with $V\ket{h}$ and denote the operator analog to $V\ket{h}\in H_B\ten H_E$ by $h$ as follows,
\[\ket{h}\in H_{A}\longleftrightarrow V\ket{h}\in H_{B}\ten H_E \longleftrightarrow  h\in \B(H_E,H_B)\]
Using this notation, we suppress the isometry $V$ and view the channel and its complementary channel as the restriction of partial traces on the Stinespring space,
\begin{align}\N(\ket{h}\bra{k})=hk^*\pl ,\pl \N^E(\ket{h}\bra{k})=k^*h\pl, \pl\text{for}\pl\ket{h},\ket{k}\in H_{A}\subset H_B\ten H_E\pl.\label{id}\end{align}
Basically, the information about isometry $V$ is encoded into its image $ran(V)\subset H_E\ten H_B$. This identification will be used to simplify our notations throughout the paper.
\subsection{TRO's and TRO channels}
Let us recall that a \emph{ternary ring of operators} (TRO) $X$ between Hilbert spaces $H$ and $K$ is a closed subspace of $\B(H,K)$ stable under the triple product
\[x,y,z\in X \Rightarrow xy^*z\in X \pl .\] A TRO $X$ is a corner of its linking $C^*$-algebra $\mathcal{A}(X)$ introduced in \cite{brown},
\[\mathcal{A}(X)=\text{span} \{\left[\begin{array}{cc}xy^*& z\\w^* &v^*u\end{array}\right]| \pl x,y,z,u,v,w\in X\}=\left[\begin{array}{cc}\mathcal{L}(X)& X\\X^* &\mathcal{R}(X)\end{array}\right]\pl.\]
The two diagonal blocks are $C^*$-algebras,
 \[\mathcal{L}(X)=\text{span}\{\pl xy^*|\pl x,y\in X\}\subset \B(K)\pl,\pl \mathcal{R}(X)=\text{span}\{\pl x^*y \pl |\pl x,y\in X\}\subset \B(H) \pl.\]
$\mathcal{L}(X)$ is called the \emph{left algebra} of $X$ and $\mathcal{R}(X)$ is called the \emph{right algebra}. They together with $\A(X)$ play an important role in the study of TROs (see again e.g. \cite{kaur}). In particular, $X$ is a natural $\mathcal{L}(X)$-$\mathcal{R}(X)$ bimodule \[\mathcal{L}(X)X=X\pl ,\pl X\mathcal{R}(X)=X\pl .\]
In finite dimensions, TRO's are direct sums of rectangular matrices with multiplicity. Namely, a TRO $X$ is isomorphic to $\displaystyle \oplus_{i}( M_{n_i,m_i}\ten  1_{l_i})$, where $l_i$ is the multiplicity of $i$th diagonal block $M_{n_i, m_i}$. In this situation,
\[\mathcal{L}(X)=\oplus_i (M_{n_i}\ten  1_{l_i})\pl , \pl \mathcal{R}(X)=\oplus_i( M_{m_i}\ten  1_{l_i})\pl.\]
In most of our discussions, the multiplicities $l_i$ are irrelevant and we may simple write $X\cong \oplus_i \p M_{n_i,m_i}$.

\begin{prop}\label{pt1} Let $\N$ be a quantum channel with its Stinespring space $X$ being a $TRO$. Then $\N$ is a direct sum of partial traces and the ranges  $ran(\N)=\mathcal{L}(X), ran(\N^E)=\mathcal{R}(X)$.
\end{prop}
\begin{proof} We can decompose $X$ as
\[{X}=\oplus_i {X_i}\pl , \pl{X_i}\cong M_{n_i,m_i}\pl.\]
 Because ${X_i}'s$ are diagonal summands that are mutually orthogonal subspaces with respect to the Hilbert-Schmidt norm, the channel $\N$ can be written as
\[\N(\ket{x}\bra{y})=xy^*=\oplus_i x_iy_i^*\pl,\pl \ket{x}=\oplus_i \ket{x_i}\pl, \ket{y}=\oplus_i \ket{y_i}     \pl,\]
where $x_i,y_i\in X_i$. It is sufficient to see on each subspace ${X_i}$, $\N$ is a partial trace. Indeed, by identifying ${X_i}\cong M_{n_i,m_i}\cong \mathbb{C}^{n_i}\ten \mathbb{C}^{m_i}$ as Hilbert spaces, we know from \eqref{ptrace} that $\N=\oplus_i \p \N_i$ and
\begin{align*}&\N_i(\ket{x_i}\bra{y_i})=x_iy_i^*=id_{n_i}\ten tr_{m_i}(\ket{x_i}\bra{y_i})\pl. \qedhere\end{align*}
\end{proof}
\begin{rem}{\rm To be precise, $X$ may be of the form $\oplus_i \p M_{n_i,m_i}\ten 1_{l_i}$ with the multiplicity $l_i$ for $i$-th block. Each direct summand $\N_i$ is a ``generalized'' partial traces as follows
\[\N_i(\rho_i)=(id_{n_i}\ten tr_{m_i}(\rho_i))\ten \omega_{l_k}\pl,\]
where $\omega_{l_k}=\frac{1}{l_k}1_{l_k}$ is the $l_k$-dimensional completely mixed state. Namely, $\N_i$ is a partial trace plus a dummy state $\omega_{l_k}$. The channel $\N=\oplus_i\N_i$ here is equivalent to the one without redundancy in Proposition \ref{pt}, in the sense that they can factor through each other. In most of situations they are equivalent and we will ignore the dummy multiplicity and use the simpler identification $X=\oplus_i \p M_{n_i,m_i}$.}
\end{rem}
Let $\tau=\frac{tr}{|H|}$ be the normalized trace on $\B(H)$. A positive operator $f$ is a normalized density if $\tau(f)=1$. Note that this normalization differs from the usual matrix trace --- for instance, the identity operator $1$ is a normalized density. This normalized trace is more natural in von Neumann algebras and will simplify our notations. Given a $C^*$-subalgebra $M\subset \B(H)$, the \emph{conditional expectation} $\E_M$ is the unique CPTP and unital map from $\B(H)$ onto $M$ (or $M+\mathbb{C}1$ if $M$ is nonunital) such that
\begin{align}\label{proj} \tau(\E_M(x)y) \lel \tau(xy) \quad \mbox{for}\quad  x\in \B(H)\pl,\pl  y\in M \pl.\end{align}
We say a positive operator $x$ is \emph{independent} of $M$ if $\E_M(x)=\tau(x) 1$, or equivalently
\[\tau(xy)=\tau(x)\tau(y) \pl, \pl  \text{for all}  \pl y\in M.\]
We say $x$ is \emph{strongly independent} of $M$ if all the powers $x^n$ are independent of $M$. The strong independence is equivalent to say that there exists a $C^*$-subalgebra $N$ such that $x\in N$ and $N$ is independent of $M$ (every element in $N$ is independent of $M$).

Now we define the modified TRO channels. Let $\N:\B(H_{A})\to \B(H_B)$ be a quantum channel with its Stinespring space $X\subset \B(H_E,H_B)$ being a TRO. Recall that with identification $H_A\cong \{\ket{x}\pl | x\in X\}$ as Hilbert spaces, $\N$ is written as
\[\N(\ket{x}\bra{y})=xy^* \pl, \pl x,y\in X \subset \B(H_E,H_B)\pl.\]
Then for any operator $f$ on $H_E$, we define the following map
\begin{align*}\N_f: \B(H_A) \to \B(H_B)\pl, \pl \N_f(\ket{x}\bra{y})=xfy^* \pl.\pl\end{align*}
Clearly, $\N_f=\N$ when $f=1$. Note that this is equivalent to the form \eqref{mchannel} in the introduction via $H_A\cong X$ as Hilbert spaces.
\begin{prop} \label{channel}Let $\N_f$ be defined as above. Suppose $f\in\B(H_E)$ is an operator independent of the right algebra $\R(X)$. Then $\E_{\mathcal{L}(X)}\circ\N_f=\tau(f)\N$.
In particular, $\N_f$ is a quantum channel if $f$ is a normalized density independent of $\R(X)$, and its Stinespring isometry is given by \[V_f:H_A\to H_B\ten H_E \pl, \pl V_f\ket{x}=\ket{x\sqrt{f}}\pl.\]
\end{prop}
\begin{proof}Let $a\in \mathcal{L}(X)$. By independence, we have that for any $\ket{x},\ket{y}$,
\begin{align*}\tau_B(a\p \N_f(\ket{x}&\bra{y}))=\frac{1}{|B|}tr_B(axfy^*)=
\frac{1}{|B|}tr_E(y^*axf)\\&=\frac{1}{|B|}\tau_E(f)tr_E(y^*ax)
=\frac{1}{|B|}\tau_E(f)tr_B(axy^*)=\tau_E(f)\tau_B(a\p \N(\ket{x}\bra{y})).
\end{align*}
Thus,
\[\E_{\mathcal{L}(X)}\circ\N_f(\ket{x}\bra{y})=\tau(f)\N(\ket{x}\bra{y}) \pl\]
holds for any rank one operator $\ket{x}\bra{y}$ and by linearity for arbitrary operators. Note that for positive $f$,
\[\N_f(\ket{x}\bra{y})=xfy^*= tr_E(\ket{x\sqrt{f}}\bra{y\sqrt{f}})
\pl.\]
Then it is sufficient to verify that $V_f$ is an isometry given $\tau_E(f)=1$. Indeed, we have
\begin{align*}\norm{\ket{x\sqrt{f}}}{}^2=tr_B(xfx^*)=tr_E(x^*xf)=\tau_E(f)tr_E(x^*x)=\norm{\ket{x}}{}^2.
\end{align*}
In the second last equality we use the assumption that $f$ is independent of $\R(X)$.
\end{proof}
We introduce the following notation for the normalized densities with the stronger independence.
\begin{defi} We say $\N$ is a \emph{TRO channel} if its Stinespring space is a TRO. We say $f\in \B(H_E)$ is a \emph{symbol} of $\N$ if $f$ is a normalized density {\bf strongly independent} of $\mathcal{R}(X)$. Then we define
\[\N_f(\ket{x}\bra{y})=xfy^*\pl, \pl x,y\in X\]
as a \emph{modified TRO channel}.
\end{defi}
\subsection{Local comparison property}
Recall that for $1\le p\le \infty$,  $\norm{a}{p}=tr(|a|^p)^{1/p}$ represents the Schatten $p$-class $S_p$ norm and $\norm{f}{p,\tau}=\tau(|f|^p)^{1/p}$ the $L_p$ norm with respect to normalized trace $\tau$. We fix the relation $\frac{1}{p}+\frac{1}{p'}=1$.
\begin{theorem}\label{pcom}Let $\N$ be a TRO channel with Stinespring space $X$. Let $f\in \B(H_E)$ be a symbol of $\N$. Then for any positive operators $\si\in \mathcal{L}(X)$ and $\rho$, 
\begin{align}&\text{i)}\pl\pl\norm{\N (\rho)}{p}\le \norm{\N_f(\rho)}{p}\le \norm{f}{p,\tau}\norm{\N (\rho)}{p}\pl, \\
&\text{ii)}\pl\pl\norm{\si^{-\frac{1}{2p'}}\N (\rho)\si^{-\frac{1}{2p'}}}{p}\le \norm{\si^{-\frac{1}{2p'}}\N_f (\rho)\si^{-\frac{1}{2p'}}}{p}\le \norm{f}{p,\tau}\norm{\si^{-\frac{1}{2p'}}\N (\rho)\si^{-\frac{1}{2p'}}}{p}\pl. \label{local} \end{align}
\end{theorem}
\begin{proof}We give the proof for ii). The argument for i) is similar and easier. Let $\E_{\mathcal{L}(X)}:\B(H_B)\to \mathcal{L}(X)$ be the conditional expectation onto $\mathcal{L}(X)$. By definition $\E_{\mathcal{L}(X)}$ is a quantum channel. From Proposition \ref{channel} and the assumption $f$ is independent to $\R(X)$,  $\E_{\mathcal{L}(X)}\circ\N_f=\N$ and $\E_{\mathcal{L}(X)}(\si)=\si$ for $\si\in \mathcal{L}(X)$, since $\mathcal{L}(X)$ is the range of $\N$. Then
the first inequality of \eqref{local} is an direct consequence of data processing inequality of R{\'e}nyi sandwiched relative entropy (see Section $3$ and e.g. \cite{MLDS}). Let $\si^{-1}$ be the inverse of $\si$ on its support. Write $\rho=\eta\eta^*$ with $\eta\in \B(H_A, X)\subset \B(H_A, H_B\ten H_E)$ for some Hilbert space $H_A$. Denote by $\hat{\eta}$ the corresponding operator of $\eta$ via $\B(H_A, H_B\ten H_E)\cong\B(H_A\ten H_E, H_B)$. We can write
\[\norm{\si^{-\frac{1}{2p'}}\N_f (\rho)\si^{-\frac{1}{2p'}}}{p}\pl =\pl  \norm{\si^{-\frac{1}{2p'}}\hat{\eta}(1_A\ten f^{\frac12})}{S_{2p}(H_A\ten H_E, H_B)}^2\pl .\]
Here $S_{2p}(H_A\ten H_E, H_B)$ is the Schatten $p$-class of operators in $\B(H_A\ten H_E, H_B)$. Thus, it is sufficient to show that
\[\norm{\si^{-\frac{1}{2p'}}\hat{\eta}(1_{A}\ten f^{\frac12})}{S_{2p}(H_A\ten H_E, H_B)} \pl\le \pl \norm{f^\frac{1}{2}}{2p,\tau}\norm{\si^{-\frac{1}{2p'}}\hat{\eta}}{S_{2p}(H_A\ten H_E, H_B)} \pl. \]
We prove it by a complex interpolation argument (see Appendix for basic information about complex interpolation). Define the norms
\[\norm{x}{p,\si}:\lel\norm{\si^{\frac{1}{p}}x}{p}\pl,\]
 and denote $\tilde{X}_{p,\si}$ as the space $H_A\ten X$ equipped with the above norm.
Theorem \ref{interpolation} in the Appendix verifies that $\tilde{X}_{p,\si}$ forms a interpolation family and in particular
\[\tilde{X}_{\si, 2p}=[\tilde{X}_\infty, \tilde{X}_{\si, 2}]_{\frac{1}{p}}. \]
Now assume that $\norm{\si^{-\frac{1}{2p'}}\hat{\eta}}{2p} \pl < 1$, we have
$\norm{\si^{\frac{1}{2p}}\xi}{2p} \pl< 1 $ where $\xi=\si^{-\frac{1}{2}}\hat{\eta} $ in $\tilde{X}\cong X\ten \B(H_A,\mathbb{C})$. Then there exists an analytic function $\xi: S=\{z|\pl 0\le Re (z)\le 1\}\to  \tilde{X}$ such that $\xi(1/p)=\xi$ and moreover
\[\norm{\xi(it)}{\infty}< 1 \pl, \pl \norm{\si^{\frac{1}{2}}\xi(1+it)}{2}< 1\pl.\]
Given this,
we define another analytic function $T(z)=\si^{\frac{z}{2}} \xi(z)(1_{A}\ten a^{pz})$, where $a=\frac{f^\frac{1}{2}}{\norm{f^\frac{1}{2}}{2p,\tau}}$. Observe that
\begin{align*}\norm{T(it)}{\infty}&=\norm{\si^{\frac{it}{2}} \xi(it)1_{A}\ten a^{itp}}{\infty}=\norm{\xi(it)}{\infty}\pl < 1,\\ \norm{T(1+it)}{2}&=\norm{\si^{\frac{1+it}{2}} \xi(1+it)(1\ten a^{p(1+it)})}{2}=\norm{\si^{\frac{1}{2}} \xi(1+it)(1\ten a^{{p}})}{2}\\&=\norm{\si^{\frac{1}{2}} \xi(1+it)}{2}\norm{ a^{p}}{2} \pl < 1\pl .\end{align*}
The last equality uses the fact $\norm{ a^{p}}{2}=1$ and the assumption $f$ is strongly independent of $\mathcal{R}(X)$.
By Stein's interpolation theorem (Theorem \ref{interpolation}), we obtain
\[\norm{T(1/p)}{2p}=\norm{\si^{\frac{1}{p}} \xi(1/p)1_{A}\ten a}{2}=\frac{1}{\norm{\sqrt{f}}{2p,\tau}}\norm{ \si^{-\frac{1}{2p'}}\hat{\eta}(1_A\ten f^{\frac12})}{2p}\le 1\pl,\]
which completes the proof.
\end{proof}
It is clear from the argument that the independence for all the powers $f^{p}$ is needed for interpolation. The above result is a local property which applies for every input $\rho$. We  naturally consider the restrictions of TRO channels on subspaces. This enables us to compare with channels whose Stinespring spaces are not necessarily TRO. Recall that we use the notation $h$ for the operators in $\B(H_E,H_B)$ corresponding to the vector $V\ket{h}\in H_B\ten H_E$.

\begin{defi} Let $\N:\B(H_{A})\to \B(H_{B})$ be a quantum channel with Stinespring space $Y$. We call a normalized density $f\in \B(H_E)$ a \emph{symbol} of $\N$ if $f$ is strongly independent of the $C^*$-algebra $\mathcal{L}(Y)$ generated by $YY^*=\{hk^*|\pl \ket{h},\ket{k}\in Y\}$. For each symbol $f$, we define the modified channel $\N_f$ as follows,
\begin{align}\label{trochannel}\N_f(\ket{x}\bra{y})=xfy^*\pl, \pl\pl \ket{x},\ket{y}\in H_A \subset H_B\ten H_E\pl.\end{align}
\end{defi}

\begin{rem}{\rm  a) Let $\mathcal{L}(Y)$ be the $C^*$-algebra generated by $\{x^*y\pl|\pl x,y\in Y\}$ and $\mathcal{R}(Y)$ be the $C^*$-algebra generated by $\{yx^*|\pl x,y\in Y\}$. Then $X=Y\mathcal{R}(Y)=\mathcal{L}(Y)Y$ is a TRO and actually the smallest TRO containing $Y$. Therefore every symbol of $\N$ gives rise to a modified  $X$-TRO channel $\M_f(|x\ran\lan y|)=xfy^*, x,y\in X$ and $\N_f$ is the restriction of $\M_f$ on $H_A\cong Y$.

b) Using this terminology every channel is a restriction of a TRO channel with a trivial symbol $1$. However, the smallest TRO obtained from the minimal Stinespring dilation may produce a large left algebra $\mathcal{L}(Y)$, which leads to ineffective capacity estimates. Our estimates in the next section are more effective when the TRO is small. }
\end{rem}

The local comparison property automatically generalizes to the restrictions onto subspaces.
\begin{cor}\label{slocal}Let $\N$ be a quantum channel and $f$ be a symbol of $\N$. Then for any positive operators $\si\in \text{Ran}(\N)$ and $\rho$,
\begin{align*}\norm{\si^{-\frac{1}{2p'}}\N (\rho)\si^{-\frac{1}{2p'}}}{p}\le \norm{\si^{-\frac{1}{2p'}}\N_f (\rho)\si^{-\frac{1}{2p'}}}{p}\le \norm{f}{p,\tau}\norm{\si^{-\frac{1}{2p'}}\N (\rho)\si^{-\frac{1}{2p'}}}{p}\pl. \end{align*}
\end{cor}
The definition of symbols is compatible with tensor products.
\begin{prop} \label{tensor}Let $\N: \B(H_A)\to \B(H_B)$ and $\M: \B(H_{A'})\to \B(H_{B'})$ be two quantum channels. Let $f$ be a symbol of $\N$ and $g$ be a symbol of $\M$. Then $f\ten g$ is a symbol of $\N \ten \M$ and
${(\N\ten \M)_{f\ten g} =\N_f\ten \M_g }$. In particular, for any channel $\M$ the identity operator $1$ is always a symbol and
$(\N\ten \M)_{f\ten 1} =\N_f\ten \M$.
\end{prop}
\begin{proof}
Let $Y^\N$ be the Stinespring spaces of $\N$ and $Y^\M$ be the Stinespring spaces of $\M$. Since $f$ and $g$ are symbols of $\N$ and $\M$ respectively, there exists TROs $X^\N\subset \B(H_E,H_B)$ containing $Y^\N$ and $X^\M\subset\B(H_{E'},H_{B'})$ containing $Y^\M$ such that $f$ is strongly independent of $\mathcal{R}(X^\N)$ and $g$ is strongly independent of $\mathcal{R}(X^\M)$. Then
\[X^\N\ten X^\M \subset \B(H_{E},H_{B})\ten \B(H_{E'},H_{B'})\cong \B(H_{E}\ten H_{E'},H_{B}\ten H_{B'})\]
is a TRO containing the Stinespring space $Y^{\N\ten\M}=Y^\N\ten Y^\M$. It is easy to see that $f\ten g\in \B(H_E)\ten \B(H_B)$ is strongly independent of $\mathcal{R}(X^\N\ten X^\M)=\mathcal{R}(X^\N)\ten\mathcal{R}(X^\M)$. Moreover, $f\ten g$ is again a normalized density hence a symbol for $\N\ten \M$. For $\ket{h_0},\ket{k_0}\in {H_A}, \ket{h_1},\ket{k_1}\in {H_{A'}}$,
\begin{align*}&(\N\ten \M)_{f\ten g}(\ket{h_0}\bra{k_0}\ten \ket{h_1}\bra{k_1})=(h_0\ten h_1)(f\ten g)(k_0\ten k_1)\\&=(h_0fk_0^*)\ten (h_1gk_1^*)=\N_f(\ket{h_0}\bra{k_0})\ten \M_g(\ket{h_1}\bra{k_1})\pl. \qedhere
\end{align*}
\end{proof}

\section{Applications to Capacity estimates}
\subsection{Entropic inequalities}
Recall that the relative entropy $D(\rho||\si)$ for two states $\rho$ and $\si$ is defined as,
\begin{align*} D(\rho||\si)=
\begin{cases}
tr(\rho\log \rho -\rho\log \si) & \text{if}\pl\pl \text{supp}(\rho)\subset\text{supp}(\si)  \\
+\infty      &  \text{else}
\end{cases}
 \pl.\end{align*}
For a bipartite state $\rho^{AB}$, the coherent information $I_c(A\ran B)_\rho$, mutual information $I(A:B)_\rho$ are given by,
\begin{align*}
 &I_c(A\ran B)_\rho=H(\rho^{AB})-H(\rho^B)=\inf_\si D(\rho^{AB}||1_A\ten \si^B) \pl , \\\pl &I(A:B)_\rho=H(\rho^{A})+H(\rho^B)-H(\rho^{AB})=\inf_\si D(\rho^{AB}||\rho^A\ten \si^B)\pl,
 \end{align*}
 where the infimum runs over all states $\si^B$ on $H_B$. The sandwiched R{\'e}nyi relative entropy $D_p(\rho||\si)$ and sandwiched R{\'e}nyi conditional entropy were introduced in \cite{MLDS,cconverse}. For ${1< p\le \infty}$, it can be written using Schatten $p$-norms as follows,
\begin{align*}&D_p(\rho||\si)=
p'\log \norm{\si^{-\frac{1}{2p'}}\rho\si^{-\frac{1}{2p'}}}{p} \pl\pl \text{(if finite)}
\pl ,\pl  \lim_{p\to 1} D_p(\rho||\si )=D(\rho||\si) \pl,\\
&H_p(A|B)_\rho=- \inf_{\si^B}D_p(\rho^{AB}||1_A\ten \si^B)\pl,\pl
\lim_{p\to 1}H_p(A|B)_\rho=H(AB)_\rho-H(B)_\rho\pl.
\end{align*}
The latter one connects to the vector-valued noncommutative $L_p$-spaces introduced by Pisier (see \cite{pvp}). Indeed, let $1\le p,q \le \infty$ and fix $1/r=|1/p-1/q|$. For a bipartite operator $\rho \in \B(H_A\ten H_B)$, the $S_q(A,S_p(B))$ norms are given as follows: for $p\le q$, $1/q+1/r=1/p$,
\begin{equation}\label{QPP}
  \|\rho\|_{S_q(A,S_p(B))} \lel
  \sup_{\|a\|_{S_{2r}(H_A)}\|b\|_{S_{2r}(H_A)}\le 1} \|(a\ten 1_B)\rho(b\ten 1_B)\|_{S_p(H_A\ten H_B)} \pl,
  \end{equation}
and for $p\ge q$, $1/p+1/r=1/q$,
 \begin{equation}\label{PQQ}
  \|\rho\|_{S_q(A,S_p(B))} \lel
  \inf_{\rho=(a\ten 1_B)\eta(b\ten 1_B)} \|a\|_{S_{2r}(H_A)}\|\eta\|_{S_p(H_A\ten H_B)}\|b\|_{S_{2r}(H_A)} \pl.
   \end{equation}
When $\rho$ is positive, it is sufficient to consider $a=b\ge 0$ in \eqref{QPP} and \eqref{PQQ}, and then the $S_1(S_p)$ norm connects to the sandwiched R{\'e}nyi conditional entropy as follows,
\begin{align*}-p'\log \|\rho\|_{S_1(B,S_p(A))}=H_p(A|B)_\rho \pl. \end{align*}
This observation enables us to translate norm estimates into entropic inequalities.
\begin{cor}\label{entropyin} Let $\N:\B(H_{A'})\to \B(H_B)$ be a channel and $f$ be a symbol of $\N$. Let $H_A$ be an arbitrary Hilbert space and $\rho^{AA'}$ be a bipartite state $H_A\ten H_{A'}$. Denote $\pl\omega^{AB}_f=id_A\ten\N_f(\rho^{AA'})$ and $\pl\omega^{AB}=id_A\ten\N(\rho^{AA'})$. Then the following inequalities hold
\begin{enumerate}
\item[i)] $H(AB)_{\omega}-\tau(f\log f)\le H(AB)_{\omega_f}\le H(AB)_{\omega}$;
\item[ii)]$I_c(A\ran B)_{\omega}\le I_c(A\ran B)_{\omega_f}\le I_c(A\ran B)_{\omega}+\tau(f \log f)$;
\item[iii)]$I(A;B)_{\omega}\le I(A;B)_{\omega_f}\le I(A;B)_{\omega}+\tau(f\log f)$.
\end{enumerate}
\end{cor}
\begin{proof}By Lemma \ref{tensor}, $f\ten 1$ is a symbol of $ \N\ten id_A$ and $(\N\ten id_A)_{f\ten 1}=\N_f\ten id_A $. The first inequality in Theorem \ref{slocal} gives $\norm{\omega}{p}\le  \norm{\omega_f}{p}\le \norm{f}{p,\tau}\norm{\omega}{p}$. Then i) follows from taking  logarithm and the limit
\begin{align*}&\lim_{p\to 1^+} p'\log \norm{\omega}{p}=-\lim_{p\to 1^+}H_p(\omega)=-H(\omega)\pl,\pl \lim_{p\to 1^+} p'\log \norm{f}{p,\tau}= \tau(f\log f)\pl.
 \end{align*}
For ii), denote $\E:=\E_{\mathcal{L}(X)}$ the conditional expectation onto the the left algebra $\mathcal{L(X)}=ran(\N)$. We have $\E\circ \N_f= \N$ by Proposition \ref{channel}. Then $(\E\ten id_A)(\omega_f)=\omega$ and the data processing inequality implies
\begin{align*}
I_c(A\ran B)_{\omega_f}\ge I_c(A\ran B)_{\omega}\pl.
\end{align*}
For the other direction, applying Theorem \ref{local},
\begin{align*}
\norm{\omega_f}{S_1(B,S_p(A))}&=\inf_{\si^B}\norm{(\si^{-\frac{1}{2p'}}\ten 1_A)\omega_f(\si^{-\frac{1}{2p'}}\ten 1_A)}{p} \\ &\le \inf_{\si^B\in \mathcal{L}(X)}\norm{(\si^{-\frac{1}{2p'}}\ten 1_A)\omega_f(\si^{-\frac{1}{2p'}}\ten 1_A)}{p}
\\&\le \norm{f}{p,\tau}\inf_{\si^B\in \mathcal{L}(X)}\norm{(\si^{-\frac{1}{2p'}}\ten 1_A)\omega(\si^{-\frac{1}{2p'}}\ten 1_A)}{p}\pl.
\end{align*}
Note that $\E(\omega)=\omega$ and by the data processing inequality,
\begin{align*}
\norm{\omega}{S_1(B,S_p(A))}&=\inf_{\si^B}\norm{(\si^{-\frac{1}{2p'}}\ten 1_A)\omega(\si^{-\frac{1}{2p'}}\ten 1_A)}{p}
\\&\ge \inf_{\si^B} \norm{(\E(\si)^{-\frac{1}{2p'}}\ten 1_A)\omega(\E(\si)^{-\frac{1}{2p'}}\ten 1_A)}{p}
\\&\ge \inf_{\si^B\in \mathcal{L}(X)} \norm{(\si^{-\frac{1}{2p'}}\ten 1_A)\omega(\si^{-\frac{1}{2p'}}\ten 1_A)}{p}
\\&\ge \inf_{\si^B}\norm{(\si^{-\frac{1}{2p'}}\ten 1_A)\omega(\si^{-\frac{1}{2p'}}\ten 1_A)}{p}\pl.
\end{align*}
Thus,
\begin{align}\norm{\omega}{S_1(B,S_p(A))}\pl\le\pl \norm{\omega_f }{S_1(B,S_p(A))}\pl\le \pl\norm{f}{p,\tau}\norm{\omega}{S_1(B,S_p(A))}. \label{conditionalp}\end{align}
We obtain ii) via the limit
\[\lim_{p\to 1^+}p'\log \|\omega\|_{S_1(B,S_p(A))}=-\lim_{p\to 1^+} H_p(A|B)_\omega= I_c(A\ran B)_\omega\pl.\]
Finally, iii) is a consequence of ii) because
$I(A:B)=H(A)+I_c(A\ran B)$ and $\omega_f^A=\omega^A$.
\end{proof}
\begin{rem}{\rm a) The term $\tau(f\log f)$ corresponds to a normalized entropy that differs from the usual entropy by a constant. Namely, $\tau(f\log f)=\log |E|-H(\frac{1}{|E|}f)$, $|E|$ is the dimension of system, $\frac{1}{|E|}f$ is a density operator of the matrix trace.

b)  The inequality $\eqref{conditionalp}$ is of its own interests. It states that for any state $\rho^{AA'}$,
\begin{align*}
\norm{\N\ten id_A (\rho)}{S_1(B,S_p(A))}\pl\le\pl \norm{\N_f\ten id_A (\rho)}{S_1(B,S_p(A))}\pl\le \pl\norm{f}{p,\tau}\norm{\N\ten id_A (\rho)}{S_1(B,S_p(A))}\pl.\end{align*}
}
\end{rem}

\subsection{Capacity Bounds}
The comparison property naturally extends to capacities of quantum channels. Let us  recall the operational definitions of channel capacities.

Let $\N:\B(H_{A'})\to \B(H_{B})$ be a quantum channel and $V\in \B(H_{A'}, H_B\ten H_E)$ be its Stinespring isometry. A quantum code $\mathcal{C}$ over $\N$ is a triple
\[\mathcal{C}=(m, \E, \D)\pl ,\]
which consists of an encoding $\E:M_{m} \to \B({H_{A'}})$ and a decoding $\D: \B(H_{B})\to M_{m} $ as completely positive trace preserving maps.  $|\mathcal{C}|=m$ is the size of the code. The quantum communication fidelity $\F_Q$ of the code $\mathcal{C}$ is defined by
\[\F_Q(\C, \N)=\bra{\psi_m} id_{m}\ten(\D\circ\N\circ\E)(\ket{\psi_m}\bra{\psi_m})\ket{\psi_m}\pl ,\]
where $\psi_m=\frac{1}{m}\sum_{i,j} e_{ij}\ten e_{ij}$ is the maximally entangled state on $M_m\ten M_m$. A rate triple $(n,R,\epsilon)$ consists of the number $n$ of channel uses, the rate $R$ of transmission and the error $\epsilon\in [0,1]$. We say a rate triple $(n,R,\epsilon)$ is achievable on $\N$ for quantum communication if there exists a quantum code $\mathcal{C}$ of $\N^{\ten n}$ such that \[\frac{\log m}{n}\ge R \pl\pl\pl\pl \text{and} \pl\pl\pl\pl \F_Q(\C, \N^{\ten n})\ge 1-\epsilon. \] Then quantum capacity $Q(\N)$ is defined as
\[Q(\N)=\lim_{\epsilon\to 0}\sup_n\{R\pl|\pl (n,R,\epsilon)\pl \text{achievable on}\pl \N \pl\text{for quantum communication}\}\pl.\]
Similarly, one can define the classical capacity $C(\N)$, private classical capacity $P(\N)$ and entanglement-assisted classical capacity $C_{EA}$. The classical capacity $C(\N)$ is the largest rate of classical bits that the channel $\N$ can reliably transmit from Alice to Bob. The private capacity $P(\N)$ is still for transmitting classical information, but which would be indiscernible to a hypothetical eavesdropper with complete access to the environment. The entanglement-assisted classical capacity $C_{EA}$ considers the improved rate with the assistance of (unlimited) pre-generated bipartite entanglement shared by the sender and receiver. We refer to \cite{Wildebook} for the formal definitions of $C, P$ and $C_{EA}$.

Thanks to the capacity theorems proved by Holevo \cite{Holevo1973,Holevo1996},  Schumacher and
Westmoreland \cite{SW}, Bennett et al \cite{BSST}, Lloyd \cite{Lloyd}, Shor \cite{Shor} and Devetak \cite{Devetak}, these operationally defined capacities are characterized by entropic expressions as follows,
\begin{align*}
&C(\N)\lel\lim_{k\to \infty} \frac{1}{k}\chi(\N^{\ten k})\pl, \pl\pl \pl\chi(\N)\lel\max_{\rho^{XA'} }I(X;B)_\omega\pl; \\
&Q(\N)\lel\lim_{k\to \infty} \frac{1}{k}Q^{(1)}(\N^{\ten k})\pl, \pl Q^{(1)}(\N)\lel\max_{\rho^{AA'} }I(A\ran B)_\omega\pl; \\
&P(\N)\lel\lim_{k\to \infty} \frac{1}{k}P^{(1)}(\N^{\ten k})\pl, \pl P^{(1)}(\N)\lel\max_{\rho^{XA'} }I(X;B)_\omega-I(X;E)_\omega\pl; \\
&C_{EA}(\N)\lel\max_{\rho^{AA'}}I(A;B)_\omega\pl,
\end{align*}
where the maximums in $C_{EA}$ and $Q^{(1)}$ run over bipartite input states $\rho^{AA'}$ and for $\chi$ and $P^{(1)}$ classical-quantum $\rho^{XA'}$. Here $\omega$ always denotes the output of $\rho$. In the four capacities above, only $C_{EA}$ admits a single-letter expression. The other three involve with the limits --the \emph{regularization} over many uses of the channel. Motivated by the super-additive phenomenon of the ``one-shot'' expressions $\chi,Q^{(1)}$ and $P^{(1)}$, Winter and Yang in \cite{WD} introduced the potential capacities $\chi^{(p)}, Q^{(p)},P^{(p)}$ as follows,
\begin{align*}\chi^{(p)}(\N)&=\sup_\M \chi(\N\ten\M)-\chi(\M) \pl, \pl Q^{(p)}(\N)=\sup_\M Q^{(1)}(\N\ten\M)-Q^{(1)}(\M)\pl, \\
P^{(p)}(\N)&=\sup_\M P^{(1)}(\N\ten\M)-P^{(1)}(\M)\pl,
\end{align*}
where the supremums runs over all channels $\M$. Note that here we use different notations from \cite{WD} to save the subscript ``$p$'' for $L_p$-norms and R{\'e}nyi-type expressions. The potential capacity is always an upper bound for corresponding capacity and hence the one-shot expression. A channel $\N$ is strongly additive for $\chi$ (resp. $Q^{(1)}$ and $P^{(1)}$) if $\chi(\N)=\chi^{(p)}(\N)$ (resp. $Q^{(1)}(\N)=Q^{(p)}(\N)$ and $P^{(1)}(\N)=P^{(p)}(\N)$). This means ${\chi(\N\ten \M)=\chi(\N)+\chi(\M)}$ (similar for $Q^{(1)}$ and $P^{(1)}$) for any $\M$ and hence $\chi(\N)=C(\N)$ (resp. $Q^{(1)}=Q$ and $P^{(1)}=P$).
\begin{prop}\label{convex}$\chi$, $Q^{(1)}$ and $P^{(1)}$ and their potential analogs are convex functions over channels.
\end{prop}
\begin{proof}We provide a uniform argument using heralded channels. Given two channels $\N:\B(H_{A'})\to \B(H_{B_1})$ and $\M: \B(H_{A'})\to \B(H_{B_2})$ with common input space, let us define the heralded channel $\Phi_\la:\B(H_{A'}) \to \B(H_{B_1}\oplus H_{B_2})$ with a probability $\la\in [0,1]$,
\[\Phi_\la(\rho)=\la \p \N(\rho)\oplus (1-\la) \M(\rho):=\left[                \begin{array}{cc}
   \la \N(\rho) & 0\\
    0& (1-\la) \M(\rho) \\
\end{array}
\right]   \pl.   \]
The output signal is heralded because Bob knows which channel is used by measuring the corresponding block. Because of the block diagonal structure, it is not hard to see that
\begin{align*}&Q^{(1)}(\Phi_\la)=\max_{\rho^{AA'}} \pl \la I_c(A\ran B_1)+(1-\la)I_c(A\ran B_2)\pl, \\&\chi(\Phi_\la)=\max_{\rho^{XA'}} \pl \la I(X; B_1)+(1-\la)I(X; B_2)\pl.
\end{align*}
Note that the complementary channel of a heralded channel is again a heralded channel of complementary channels, i.e. $\Phi^E_\la (\rho)=\la\p \N^E(\rho)\oplus (1-\la)\M^E(\rho)$. Then a similar formula holds for one-shot private capacity $P^{(1)}$,
\[P^{(1)}(\Phi_\la)=\max_{\rho^{XA'}} \pl\la (I(X; B_1)-I(X;E_1))+(1-\la)(I(X; B_2)-I(X;E_2)) \pl.\]
Now if $\N$ and $\M$ have the same output space $H_{B_1}=H_{B_2}$, then the convex combination $\la\N +(1-\la)\M$ can be factorized through the heralded channel $\Phi_\la$ via a partial trace map. Therefore by data processing,
\begin{align*}Q^{(1)}(\la\N +(1-\la)\M)\le Q^{(1)}(\Phi_\la)&=\max_{\rho^{AA'}} \pl\la I_c(A\ran B_1)+(1-\la)I_c(A\ran B_2) \\&\le \la Q^{(1)}(\N)+(1-\la) Q^{(1)}(\M) \pl.\end{align*}
Here the $Q^{(1)}$ can be replaced by $\chi$ and $P^{(1)}$. Moreover, the convexity of potential capacities follow from the convexity of their ``one-shot'' expressions.
\end{proof}

We have seen that when the Stinespring space is TRO, the channel is a diagonal sum of partial traces. The capacity formulae of these channel follows from Proposition $1$ in \cite{wolf}.
\begin{prop}\label{pt} Let $\N=\oplus_i id_{n_k}\ten tr_{m_k}$ be a direct sum of partial traces. Then $\N$ is strongly additive for $\chi,Q^{(1)}$ and $P^{(1)}$, and moreover
\[Q^{(1)}(\N)=P^{(1)}(\N)=\log \big(\max_i n_i \big)\pl, \chi(\N)= \log ( \sum_{i}n_i)\pl, \pl C_{EA}(\N)=\log (\sum_{i} n_i^2)\pl .\]
\end{prop}
The next theorem provides the comparison property for capacities.
\begin{cor} \label{comparison}Let $\N$ be a channel and $f$ be a symbol for $\N$. Then,
\begin{enumerate}
\item[i)]$C(\N)\le C(\N_f)\le C(\N)+\tau(f\log f)$;
\item[ii)] $Q(\N)\le Q(\N_f)\le Q(\N)+\tau(f\log f)$;
 \item[iii)]$P(\N)\le P(\N_f)\le P(\N)+\tau(f\log f)$;
 \item[iv)] $C_{EA}(\N)\le C_{EA}(\N_f)\le C_{EA}(\N)+\tau( f\log f)\pl.$
 \end{enumerate}
 For i),ii) and iii), the capacity can be replaced by corresponding one-shot expression and potential capacity.
\end{cor}
\begin{proof} The inequalities for $\chi$, $Q^{(1)}$ and $C_{EA}$ follows from Corollary \ref{entropyin} by taking maximum over all possible inputs. Note that the ``one-shot'' private capacity can be rewritten as \[P^{(1)}(\N)\lel\max_{\rho^{XA'A} } \pl I_c(A\ran B)_\omega-\sum_{x}p(x)I_c(A\ran B)_{\omega_x}\pl,\]
where the maximum runs over all states \[\rho^{XA'A}=
\sum_x p(x)\ket{x}\bra{x}\ten \rho_x^{AA'}\]
and $\rho_x^{AA'}$ are pure states. The coherent information is for the output $\omega_x^{AB}=id_A\ten \N(\rho_x^{AA'})$ and $\omega^{AB}=id_A\ten \N(\tilde{\rho}^{AA'})$ where $\tilde{\rho}^{AA'}$ is any purification of $\rho^A$ (so $\tilde{\rho}^{AA'}$ may not be the reduced density of $\rho^{XA'A}$). Applying Corollary \ref{entropyin} ii) one have
\begin{align*}I_c(A\ran B)_{\omega_f}-\sum_x p(x)I_c(A\ran B)_{\omega_{f,x}}\le I_c(A\ran B)_{\omega}-\sum_x p(x)I_c(A\ran B)_{\omega_{x}}+\tau(f\log f)\pl .
\end{align*}
Then the upper bound of $P^{(1)}(\N_f)$ follows
and the lower bound is a consequence of the lifting property $\E_{\mathcal{L}(X)}\circ\N_f=\N$.
For the regularization, note that by Lemma \ref{tensor}, $f^{\ten k}$ is a symbol of $\N^{\ten k}$. Therefore  we have
\begin{align*}P(&\N_f)=\lim_{k\to\infty}\frac1k P^{(1)}\Big((\N_f)^{\ten k}\Big)=\lim_{k\to\infty}\frac1k P^{(1)}(\N^{\ten k}_{f^{\ten k}})\le\lim_{k\to\infty}\frac1k[P^{(1)}(\N^{\ten k})+\tau^k(f^{\ten k}\log f^{\ten k})]\\=&\lim_{k\to\infty}\frac1k(P^{(1)}(\N^{\ten k})+k\tau(f\log f))=\lim_{k\to\infty}\frac1k P^{(1)}(\N^{\ten k})+\tau(f\log f)=P(\N)+\tau(f\log f)\pl.
\end{align*}
Similarly, for the potential capacities, we use that $\M\ten \N_f=(\M\ten \N)_{1\ten f}$ for an arbitrary channel $\M$ and $\tau((1\ten f) \log 1(\ten f))=\tau(f\log f)$. The arguments for classical capacity and quantum capacity are the same.
\qd
The gap of upper and lower estimates are bounded uniformly by the term $\tau(f \log f)$. This can be viewed as a ``first order'' approximation of the capacity of $\N_f$ by the entropic term $\tau(f \log f)=\log |E|-H(\frac{1}{|E|}f)$.

The next theorem is a formula of the negative $cb$-entropy. The negative $cb$-entropy $-S_{cb}(\N)$ of a channel $\N: \B(H_{A'})\to \B(H_{B})$ is defined as
\begin{align*}
-S_{cb}(\N)=\sup_{\rho}{H(A)_\omega-H(AB)_\omega}\pl,
\end{align*}
where $\omega^{AB}=id_A\ten \N(\rho^{AA'})$ and the supremum runs over all pure bipartite states $\rho^{AA'}$. It was characterized in \cite{DJKR} as the derivative at $p=1$ of the completely bounded norm from trace class to Schatten $p$ class,
\begin{equation}\label{djkr}
-S_{cb}(\M) \lel \frac{d}{dp}|_{p=1}{\norm{\M:S_1(H_{A'})\rra S_p(H_B)}{cb}} \pl ,
\end{equation}
 and later rediscovered in \cite{invcoh} as
``reverse coherent information'' with an operational meaning. Recall that $|A|:=dim H_A$ denotes the dimension of a Hilbert space.
\begin{theorem}\label{negativecb}
Let $\N$ be a quantum channel and $f$ be a symbol of $\N$. Suppose that the complimentary channel $\N^E:\B(H_{A'}) \to \B(H_E)$ is unital up to a scalar, i.e. ${\N^E(1_{A'})=\frac{|A'|}{|E|}1_E}$. Then
\[-S_{cb}(\N_f)=\log\frac{|A'|}{|E|}+\tau(f\log f)\pl.\]
\end{theorem}
\begin{proof} Let $H_A\cong {H_{A'}}$ and denote $e_{ij}$ the matrix units in $\B(H_A)\cong M_m$ for $m=|A'|=|A|$. Let $\{\ket{h_i}\}$ be an orthonormal basis of ${H_{A'}}$. For a channel $\M: \B(H_{A'})\to \B(H_B)$, its Choi matrix $J_\M$ is given by
\[J_\M=\sum_{i,j}e_{ij}\ten \M(\ket{h_i}\bra{h_j}) \in \B(H_A\ten H_B).\]
The completely bounded $1\to p$ norm of a map $\M$ is same with the vector-valued $(\infty,p)$ norm (defined in \eqref{PQQ}) of its Choi matrix $J_\M$ (see e.g. \cite{Psbook,ER}),
\[\norm{\M:S_1(H_{A'})\rra S_p(H_B)}{cb}=\norm{J_\M}{S_\infty(A,S_p(B))}\pl.\]
In particular, for $p=\infty$, $S_\infty(A,S_\infty(B))=\B(H_A\ten H_B)$.
The Choi matrix of $\N_f$ is given by
\[\sum_{ij}e_{ij}\ten \N_f(\ket{h_i}\bra{h_j})=\sum_{i,j}e_{i,j}\ten (h_ifh_j^*)=W(e_{11}\ten f)W^*\in  \B(H_A\ten H_B)\pl,\]
where $h_i$ are operators in $\B(H_E,H_B)$ corresponding to $\ket{h_i}$ and $W=\sum_i e_{i1}\ten h_i$. Since $\N^E$ is unital up to a factor,
\[\N^E(\sum_i \ket{h_i}\bra{h_i})=\sum_i h_i^*h_i=\frac{|A'|}{|E|}1_E\pl.\]
This implies that $(\frac{|E|}{|A'|})^{\frac{1}{2}}W$ is an isometry. We then define the following $*$-homomorphism
\begin{align}\pi: \B(H_E) \to \B(H_A\ten H_B) \pl, \pl \pi(f):=\frac{|E|}{|A'|}J_{\N_f}=\frac{|E|}{|A'|}W(e_{11}\ten f)W^*\pl.\label{pinfty}\end{align}
Note that $tr(f)=tr(\pi(f))$, then
\begin{align*}  \|f\|_{p,\tau}^p \lel |E|^{-1}tr_E(|f|^p) \lel
|E|^{-1}tr_{AB}(\pi(|f|^p)) \lel |E|^{-1}\|\pi(|f|)\|_{p}^p \pl .\end{align*}
Therefore we get
 \begin{align*}
   &\|f\|_{p,\tau} \lel |E|^{-1/p}\|\pi(f)\|_p
   \lel |E|^{-1/p}\frac{|E|}{|A'|}\|W(e_{11}\ten f)W^*\|_{p}=
 |E|^{1-1/p}|A|^{-1}
    \|J_{\N_f}\|_{p}\pl.
   \end{align*}
By the definition \eqref{QPP}, we obtain a lower bound for the $S_\infty({A'},S_p(B))$ norm,
 \begin{equation} \label{lower}
  \|J_{\N_f}\|_{S_\infty({A'},S_p(B))}\gl |A'|^{-1/p}
 \|J_{\N_f}\|_{p}
 \lel \frac{|A'|}{|E|}^{1-1/p} \|f\|_{p,\tau}\pl .
 \end{equation}
For the upper bound, note that $\pi$ is a $^*$-homomorphism, then \begin{align}\norm{\pi(f)}{\infty}=\norm{\frac{|E|}{|A'|}J_{\N_f}}{\infty}=\frac{|E|}{|A'|}\norm{\N_f:S_1(H_{A'})\to \B(H_B)}{cb}\pl.
\end{align}
Now assume that $X$ is a TRO containing $\N$'s Stinespring space and $f$ is strongly independent $\mathcal{R}(X)$. Let $M\subset \B(H_E)$ be the $C^*$-subalgebra generated by $f$. Then for any operator $g\in M$, the map $\N_g(\ket{h}\bra{k})=hgk^*$ satisfies that
\[tr(\N_g(\rho))=\tau(g)tr(\N(\rho))\pl.\]
Thus $\norm{\N_g:S_1(H_{A'}) \to S_1(H_B)}{cb}=|\tau(g)|\le \norm{g}{1,\tau}$. we have \[\|\pi:L_1(M,\tau)\to S_\infty (A, S_1(B))\|\le \frac{|E|}{|A'|}\pl.\]
Note that for $L_\infty$ spaces, \[\|\pi:M\to \B(H_A\ten H_B)\|\le 1\pl,\] because $\pi$ is a $*$-homomorphism. Then by Stein's interpolation theorem (Theorem \ref{interpolation}),
  \begin{align}\label{upth} \|\pi:L_p(M,\tau)\to S_\infty (A, S_p(B))\|\kl (\frac{|E|}{|A'|})^{1/p} \pl ,\end{align}
Combining \eqref{upth} with \eqref{lower}, the upper and lower bounds coincide and give
 \begin{align*}\|J_{\N_f}\|_{S_\infty (A, S_p(B))} \lel  (\frac{|A'|}{|E|})^{1-1/p}\|f\|_{p,\tau} \pl .\end{align*}
The assertion follows by differentiating the above equality at $p=1$. Note that for all $1\le p\le \infty$, the maximal entangled state $\sum_{i}e_{ij}\ten\ket{h_i}\bra{h_j}$ is a norm attaining element.
\end{proof}
\subsection{The capacity regions}
\noindent
The capacity regions of a quantum channel consider the trade offs between different resources in quantum information theory. The notion of a capacity region relies on the availability of quantum protocols, such as teleportation and dense coding, that exchange one type of resource for another. Based on research due to Devetak and Shor \cite{DS}, Abeyesinghe et al \cite{mother}, Collins and Popescu \cite{CP} and many others, Hsieh and Wilde introduced the two kinds of capacity regions: the quantum dynamic region $C_{CQE}$ and private dynamic region $C_{RPS}$. The quantum dynamic region $C_{CQE}$ considers a combined version of classical communication ``$C$'', quantum communication ``$Q$'' and entanglement generation ``$E$'', while the private dynamic region $C_{RPS}$, with the idea of the Collins-Popescu analogy \cite{CP},  unifies the public classical communication ``$R$'', private classical communication ``$P$'' and secret key distribution ``$S$''. We refer to their papers \cite{HW1,HW2} for the operational definitions of $C_{CQE}$ and $C_{RPS}$. Here we state the capacity region theorems from \cite{HW1,HW2}.

Let $\N: \B(H_{A'})\to \B(H_B)$ be a quantum channel and ${V:H_{A'}\to H_B\ten H_E}$ be its Stinespring isometry. The quantum dynamic region $C_{CQE}(\N)$ is characterized as follows,
 \[
C_{CQE}(\N)=\overline{\bigcup^{\infty}_{k=1}\frac{1}{k}C_{CQE}^{(1)}(\N^{\otimes k})}\ , \ \ \ \ \ C_{CQE}^{(1)}\equiv \bigcup_\omega C^{(1)}_{CQE,\omega}\pl,
\]
where the overbar represents the closure of a set. The ``one-shot" region $C_{CQE}^{(1)}\subset \mathbb{R}^3$ is the union of the ``one-shot, one-state" regions $C_{CQE,\omega}^{(1)}$, which are the sets of all rate triples $(C,Q,E)$ such that:
\begin{align*}
C+2Q\le I(AX;B)_\omega\pl, \ Q+E\le I(A\rangle BX)_\omega\pl , \  C+Q+E&\le I(X;B)_\omega+I(A\rangle BX)_\omega\pl.
\end{align*}
The above entropy quantities are with respect to a classical-quantum state
\[\omega^{XABE}\lel\sum_x p_X(x)\ketbra{x}^X\otimes(1_A \ten V)\rho^{AA'}_x(1_A\ten V^*)\]
and the states $\rho^{AA'}_x$ are pure. Similarly, the private dynamic region is given by,
\[
C_{RPS}(\N)=\overline{\bigcup^{\infty}_{k=1}\frac{1}{k}C_{RPS}^{(1)}(\N^{\otimes})}\pl, \pl C_{RPS}^{(1)}\equiv \bigcup_\omega C^{(1)}_{RPS,\omega} \pl.\]
The ``one-shot, one-state" region $C^{(1)}_{RPS}(\N)\subset \mathbb{R}^3$ is the set of all triples $(R,P,S)$ such that
\begin{align*}
R&+P\le I(YX;B)_\omega\pl ,
\pl P+S\le I(Y;B|X)_\omega-I(Y;E|X)_\omega\pl ,\\
R&+P+S\le I(YX;B)_\omega-I(Y;E|X)_\omega\pl.
\end{align*}
The above entropic quantities are with respect to a classical-quantum state $\omega^{XYBE}$ where
\[\omega^{XYBE}\equiv\sum_x p_{X,Y}(x,y)\ketbra{x}^X\otimes \ketbra{y}^Y \otimes (V\rho^{A'}_{x,y}V^*)\pl.\]
\begin{exam}{\rm Let $\N$ be a channel and its Stinespring space be a TRO $X\cong \oplus_i M_{n_i,m_i}$. We know from Proposition \ref{pt} that
\[\N=\oplus_i id_{n_i}\ten tr_{m_i}\] as a direct sum of partial traces. The capacity regions of this class of channels are accessible. The quantum dynamic region regularizes
$C_{CQE}(\N)=C^{(1)}_{CQE}(\N)$, and it is characterized as a union of the following regions
\begin{align*}
C+2Q&\le H(\{p_{\la,\mu}(i)\})+2\sum_i p_{\la,\mu}(i) \log n_i \pl,
\pl Q+E\le \sum_i p_{\la,\mu}(i) \log n_i\pl ,\\
C+Q+E&\le H(\{p_{\la,\mu}(i)\})+\sum_i p_{\la,\mu}(i)\log n_i\pl.
\end{align*}
for all $\la,\mu\ge0$. Here $\{p_{\la,\mu}(i)\}$ is the probability distribution given by
\[p_{\la,\mu}(i)=n_i^{\frac{2+\la+\mu}{1+\mu}}/(\sum_i n_i^{\frac{2+\la+\mu}{1+\mu}})\pl.\]
Similarly, for the public-private dynamic region, $C_{RPS}(\N)=C_{RPS}^{(1)}(\N)$ is the union of
\begin{align*}
&R+P\le H(\{q_{\la,\mu}(i)\})+\sum_i q_{\la,\mu}(i) \log n_i \pl,
P+S\le \sum_i q_{\la,\mu}(i) \log n_i\pl ,\\
&R+P+S\le H(\{q_{\la,\mu}(i)\})+\sum_i q_{\la,\mu}(i)\pl.
\end{align*}
for all $\la,\mu\ge0$. Here $\{q_{\la,\mu}(i)\}$ is the probability distribution given by
\[q_{\la,\mu}(i)=n_i^{\frac{1+\la+\mu}{1+\mu}}/(\sum_i n_i^{\frac{1+\la+\mu}{1+\mu}})\pl.\]
}
\end{exam}

In general it is difficult to completely characterize the capacity regions. Let us consider two cones,
 \[ W_1 \lel \{(C,Q,E)|2Q+C\le 0,\ Q+E\le 0,\ Q+E+C\le 0\}\]
 and \[W_2=\{(R,P,S)|\ R+P\le 0, P+S \le 0, R+P+S\le 0\}\pl.\]
 The first one is the resource trading off via teleportation, superdense coding and entanglement distribution and the second
 is the cone obtained from secret key distribution, the one-time pad and private-to-public transmission (see \cite{HW1,HW2}). We have a comparison property of the rate triple $(I(X;B)_{\omega},\frac12I(A;B|X)_{\omega},-\frac{1}{2}I(A;E|X)_{\omega} )$ for each single input state $\rho^{XA'}$ and respectively $(I(X,B)_{\omega},I(Y;B|X)_{\omega},-I(Y;E|X)_{\omega} )$ for each $\rho^{XYA'}$.

\begin{cor}Let $\N$ be a channel and $f$ be a symbol for $\N$. Denote the quantity $\tau :=\tau(f\log f)$. Then
\begin{enumerate} \item[i)] $C_{CQE}(\N)\subset C_{CQE}(\N_f) \subset C_{CQE}(\N)+(\tau,\frac{\tau}{2},\frac{\tau}{2})\pl; $ \item[ii)]$C_{RPS}(\N)\subset C_{RPS}(\N_f) \subset C_{RPS}(\N)+(\tau,\tau,\tau)\pl .$
\end{enumerate}
\end{cor}
\begin{proof}
The argument for the two kinds of regions are similar. Here we give the proof for the private dynamic region $C_{RPS}$ and the argument proof for quantum dynamic region $C_{CQE}$ is similar. Let us assume
\[\omega_f^{XYABE}=\sum_x p_{X,Y}(x,y)\ketbra{x}^X\otimes \ketbra{y}^Y \ten (1\ten V_f)(\rho^{A'A}_{x,y})(1\ten V_f^*)\pl ,\]
where $\rho^{A'A}_{x,y}$ are pure states. We denote $(R^f, P^f, S^f)$ for the rate triple \[(I(X;B)_{{\omega_f}},I(Y;B|X)_{{\omega_f}},-I(Y;E|X)_{{\omega_f}} ) \pl.\]
By the entropic inequality \eqref{entropyin},
\begin{align*} I(X;B)_{\omega_f}\le\tau(f\log f)+ I(X;B)_{\omega_1}  \pl .
\end{align*}
From this, we may assume
$R^f=R^1+\tau(f\log f)-\al_1$ for some $\al_1\gl 0 $. Similarly, we have
 \begin{align*} I(Y;B|X)_{\omega_f} &\lel H(Y|X)_{\omega_f}+ H(B|X)_{\omega_f}-H(YB|X)_{\omega_f}\\&=H(Y|X)_{\omega_f}+\sum_x p(x) H(\omega^B_{f,x})
 - \sum_x p(x)(H(Y|X=x)+\sum_y p(y|x)H(\omega^{B}_{x,y,f}))\\ & \le I(Y;B|X)_{\omega}+\tau(f\log f)\pl ,
 \end{align*}
and
 \begin{align*} &I(Y;E|X)_{\omega_f}\lel H(Y|X)_{\omega_f}+ H(E|X)_{\si,f}-H(YE|X)_{\omega_f}\\&=H(Y|X)_{\omega_f}+\sum_x p(x) H(\omega^E_{f,x})
 - \sum_x p(x)(H(Y|X=x)+\sum_y p(y|x)H(\omega^{E}_{x,y,f}))\\&
  \ge I(Y;E|X)_{\omega}-\tau(f\log f)\pl .\end{align*}
This means
 \[ P^f \lel P^1+ \tau(f\log f)-\al_2 \pl ,\pl S^f \lel S^1+ \tau(f\log f)-\al_3 \pl \]
for some $\al_2,\al_3\gl 0$. Now it is obvious that $(-\al_1,-\al_2, -\al_3)\in W$, then we have
 \[ (R^f,P^f,S^f) \in (\tau,\tau,\tau)+ (R^1,P^1,S^1)+W_2 \pl .\]
Taking the union for all $\omega$, we have
 \[ C^{(1)}_{RPS}(\N_f)\in (\tau,\tau,\tau) + C^{(1)}_{RPS}(\N)
  + W_2+W_2 \pl .\]
For the cone $W_2$, we have $W_2+W_2=W_2$ and this concludes that \[C_{RPS}^{(1)}(\N_f)\subset (\tau,\tau,\tau)+C_{RPS}^{(1)}(\N)\pl. \]
For regularization, we apply the above estimates to the tensor product channel
\begin{align*}\frac{1}{k}C^{(1)}_{RPS}(\N_f^{\ten k})&=\frac{1}{k}C^{(1)}_{RPS}\Big((\N^{\ten k})_{\pl f^{\ten k}}\Big)\subset \frac{1}{k}\Big(k(\tau,\tau,\tau)+C_{RPS}^{(1)}(\N^{\otimes k})\Big)\\&=(\tau,\tau,\tau)+\frac{1}{k}C_{RPS}^{(1)}(\N^{\otimes k})\pl,\end{align*}
which completes the proof.
\end{proof}
\subsection{Strong converse rates}A ``strong converse'' means there is a sharp drop off for code fidelity above the optimal transmission rate. More generally, we will investigate rates above which the transmission only succeeds with arbitrarily small probability. We say $r$ is a strong converse rate for quantum communication if for every sequence of achievable triple $(n,R_n,\e_n)$ of quantum communication, we have
\[\liminf_{n\to\infty}R_n>r\pl \Rightarrow \pl \lim_{n \to\infty}\e_n=1\pl.\]
The strong converse rate of classical communication and private classical communication can be defined similarly. We refer to \cite{cconverse,pconverse} for formal definitions of these two because they not used directly in this paper. The strong converse classical capacity $C^\dag$, the strong converse quantum capacity $Q^\dag$ and the strong converse private capacity $P^\dag$ are defined as the infimum of corresponding strong converse rates. We say a channel $\N$ has strong converse if the capacity equals to the strong converse capacity (respectively, $C^{\dag}(\N)=C(\N)$, ($Q^{\dag}(\N)=Q(\N)$, $P^{\dag}(\N)=P(\N)$).

There are known upper bounds for strong converse capacities. It is shown by Wilde et al \cite{cconverse} that for any channel $\N$,
\begin{align}\label{alholevo}C^\dag(\N)\le \lim_{k\to \infty}\frac{\chi_p(\N^{\ten k})}{k}\pl , \pl \chi_p(\N)=\max_{\rho^{XA'}}I_p(X;B)_\omega \pl ,\end{align}
where the sandwich R{\'e}nyi mutual information is given by
\begin{align} I_p(A;B)_\rho=\inf_{\si^B} D_p(\rho^{AB}||\rho^A\ten \si^B)\pl.
\end{align}
For the quantum strong converse, the Rains information of a quantum channel is shown to be a strong converse rate \cite{qconverse}. The relative entropy of entanglement $E_R(\N)$ is an upper bound for the private capacity \cite{PLOB15} and the strong converse capacity \cite{pconverse},
\begin{align}\label{Pp}P^\dag(\N)\le  E_R(\N)\pl, \pl E_R(\N)=\max_{\rho^{AA'}}E_R(id_A\ten \N_f(\rho))\pl.\end{align}
The relative entropy of entanglement $E_R(\rho)$ for a bipartite $\rho^{AB}$ is
\[E_R(\rho^{AB})=\inf_{\si^{AB} \in S(A:B)} D(\rho^{AB}||\si^{AB}) \pl,\]
where $S(A$:$B)$ stands for the separable states between $A$ and $B$.
These results in particular imply the strong converses of Hadamard channels and entanglement-breaking channels for classical, quantum and private communication. In their arguments, the sandwich R{\'e}nyi relative entropy plays an important role.

For $1<p\le\infty$ and $1/p+1/p'=1$,
we consider the R{\'e}nyi coherent information of a channel for as an analog of \eqref{alholevo}
\[Q_p^{(1)}(\N)=\max_{\omega=id\ten \N(\rho)} I_{c,p}(A\ran B)_\omega\pl ,\pl I_{c,p}(A\ran B)_\omega=p'\log \norm{\omega^{BA}}{S_1(B,S_p(A))}\pl. \]
The following is a folklore result which probably known to experts but not stated explicitly in the literature.
\begin{prop}\label{Qp}For any channel $\N$ and all $1<p\le\infty$.,
\[\limsup_{k\to \infty} \frac{Q_p^{(1)}(\N^{\ten k})}{k},\]
is a strong converse rate of $\N$ for quantum communication 
\end{prop}
\begin{proof} Denote $R_p=\limsup_{k\to \infty} \frac{1}{k}Q_p^{(1)}(\N^{\ten k})$.
Let $m=2^M$ and $\fs{1}{p}+\fs{1}{p'}=1$. It is sufficient to show that for an arbitrary code $\C=(m,\E,\D)$ of $\N$,
\begin{align}\F(\C,\N)\le m^{-\fs{1}{p'}}\exp(\fs{1}{p'}Q^{(1)}_p(\N))\pl \label{fs} .\end{align}
Indeed, let $\C_n$ be a sequence of codes such that $\liminf_{n\to \infty} \fs{1}{n}\log |\C_n|> R_p+\epsilon$,
\begin{align*}\F(\C_n,\N^{\ten n})\le m^{-\fs{1}{p'}}\exp(\frac{1}{p'}Q^{(1)}_p(\N^{\ten n}))
=\exp(\fs{1}{p'}(Q^{(1)}_p(\N^{\ten n})-M))
\le 2^{(-\fs{1}{p'}n\epsilon)} \pl ,
\end{align*}
for $n$ large enough. To prove \eqref{fs}, we define $\omega =id_m\ten(\D\circ\N\circ\C)(\ket{\psi_m}\bra{\psi_m})$. Then the fidelity is given by
\[\F(\C, \N)=tr(\omega\ket{\psi_m}\bra{\psi_m})\le \norm{\omega}{S_1^m(S_p^m)}\norm{\ket{\psi_m}\bra{\psi_m}}{M_m(S_{p'}^m)}\pl .\]
Note that $\E$ and $\D$ are completely positive trace preserving maps and hence
\[\norm{\E: S_1^m\to S_1({H_{A'}})}{cb}=1\pl, \pl\norm{\D: S_1(H_B)\to S_1^m}{cb}=1 \pl.\]
This implies
\[\norm{\omega}{S_1^m(S_p^m)}\le 2^{{(p'Q^{(1)}_p(\N))}} \norm{\ketbra{\psi_m}}{S_1^m\ten S_1^m}=2^{(p'Q^{(1)}_p(\N))} \pl.\]
For the second term, we use the interpolation relation (see Appendix)
\[M_m(S_{p'}^m)=[M_m(S_1^m),M_m(M_m)]_{
\frac{1}{p'}}\pl.\]
We have
\begin{align}\norm{\ketbra{\psi_m}}{M_m(S_{p'}^m)}=\norm{\ketbra{\psi_m}}{M_m( M_m)}^{\frac{1}{p}} \norm{\ketbra{\psi_m} }{M_m(S_{1}^m)}^{\frac{1}{p'}}
\le {m}^{-\fs{1}{p'}}\pl , \label{me}
\end{align}
where we used the fact $\norm{\ketbra{\psi_m} }{M_m(S_{1}^m)}=1/m$.
\eqref{me} is indeed an equality. Combining these two estimates, we obtain \eqref{fs}.
\end{proof}
\begin{exam}{\label{cd}\rm We consider again the case when the Stinespring space $X^\N$ is a TRO space. Assume that $\N=\oplus_i id_{n_i}\ten tr_{m_i}$ is a direct sum of partial traces, it is not hard to calculate that
\[\chi_p(\N)=\log \big( \sum_{i}n_i\big)\pl ,\pl Q^{(1)}_p(\N)=E_{R,p}(\N)=\log \big(\max_i n_i\big)\pl.\]
Note that all these terms are additive. Let $\M=\oplus_j id_{n'_j}\ten tr_{m'_j}$ be another direct sum of partial traces. Then
\[\N\ten \M= \oplus_{i,j} id_{n_in'_j}\ten tr_{m_im'_j}\pl.\]
is again of orthogonal sum of partial traces. Apply the above formulae for $\N\ten \M$, we obtain
\begin{align*}&\chi_p(\N\ten \M)=\log ( \sum_{i,j} n_in'_j)=\log( \sum_{i}n_i)+\log( \sum_{j}n'_j)=\chi_p(\N)+\chi_p(\M) \pl,\\
&Q_p^{(1)}(\N\ten \M)=\max_{i,j}
\log n_in'_j=\max_{i}
\log n_i+\max_{j}
\log n'_j=Q_p^{(1)}(\N)+Q_p^{(1)}(\M)\pl,
\end{align*}
and similarly for $E_{R,p}$. Hence the regularization is trivial. By Proposition \ref{pt},
TRO channels has strong converse for classical, quantum and private communication.
}
\end{exam}
The following lemma is an analog of \eqref{entropyin} for R{\'e}nyi information measures.
\begin{lemma}\label{pentropyin}Let $\N$ be a channel and $f$ be a symbol of $\N$. Let $H_A$ be an arbitrary Hilbert space and $\rho^{AA'}$ be a bipartite state $H_A\ten H_{A'}$. Denote $\pl\omega^{AB}_f=id_A\ten\N_f(\rho^{AA'})$ and $\pl\omega^{AB}=id_A\ten\N(\rho^{AA'})$. Then the following inequalities hold:\begin{enumerate}
\item[i)]$I_{c,p}(A\ran B)_{\omega_1}\le I_{c,p}(A\ran B)_{\omega_f}\le I_{c,p}(A\ran B)_{\omega_1}+p'\log \norm{f}{p,\tau}\pl;$
\item[ii)]$I_p(A:B)_{\omega_1}\le I_p(A:B)_{\omega_f}\le I_p(A:B)_{\omega_1}+p'\log \norm{f}{p,\tau}\pl; $
\item[iii)]$E_{R,p}(\omega_1)\le E_{R,p}(\omega_f)\le E_{R,p}(\omega_1)+p'\log \norm{f}{p,\tau}\pl .$
\end{enumerate}
\end{lemma}
\begin{proof} Let $X$ be a TRO containing $\N$'s Stinespring space and $f$ is strongly independent of $\mathcal{R}(X)$. All lower bounds follows from the factorization  property $\E_{\mathcal{L}(X)}\circ \N_f=\N$, where $\E_{\mathcal{L}(X)}:\B(H_B)\to \mathcal{L}(X)$ is the conditional expectation onto the left algebra $\mathcal{L}(X)$. The upper estimate of
i) is a direct consequence of the vector-valued $(1,p)$ norm inequality \eqref{conditionalp}. Indeed,
\begin{align*}
I_{c, p}(A\ran B)_{\omega_f}&=p'\log \norm{\omega_f}{S_1(B,S_p(A))}
\le p'\log \norm{f}{\tau,p}\norm{\omega}{S_1(B,S_p(A))}
\\&\le p'\log \norm{f}{\tau,p}+p'\log\norm{\omega}{S_1(B,S_p(A))}
\le p'\log \norm{f}{\tau,p}+I_{c, p}(A\ran B)_{\omega}\pl.
\end{align*}
For ii), note that $\E_{\mathcal{L}(X)}\circ \N=\N$,
\[id_A\ten \E_{\mathcal{L}(X)}(\omega)=id_A\ten (\E_{\mathcal{L}(X)}\circ\N)(\rho)=id_A\ten \N(\rho)=\omega\pl.\]
Therefore for the R{\'e}nyi mutual information,
\begin{align*}
I_p(A;B)_{\omega}&=\inf_{\si^B} D_p(\omega||\omega^{A}\ten \si^B)\ge  \inf_{\si^B} D_p(\omega||\omega^{A}\ten \E_{\mathcal{L}(X)}(\si^B))
\\ &\ge \inf_{\si^B\in \mathcal{L}(X)} D_p(\omega||\omega^{A}\ten \si^B)
\ge \inf_{\si^B} D_p(\omega||\omega^{A}\ten \si^B)\pl.
\end{align*}
Hence $I_p(A;B)_{\omega}=\inf_{\si^B\in \mathcal{L}(X)} D_p(\omega||\omega^{A}\ten \si^B)$ where it suffices to consider $\si^B\in \mathcal{L}(X)$ for the infimum. Combined with the Theorem \eqref{pcom}, we have
\begin{align*}I_p(A;B)_{\omega_f}&=\inf_{\si^B} D_p(\omega_f||\omega^{A}\ten \si^B)
\le \inf_{\si^B\in \mathcal{L}(X)} D_p(\omega_f||\omega^{A}\ten \si^B)
\\&\le\inf_{\si^B\in \mathcal{L}(X)} D_p(\omega||\omega^{A}\ten \si^B) +p'\log \norm{f}{\tau ,p}
=I_p(A;B)_{\omega}  +p'\log \norm{f}{\tau ,p}\pl.
\end{align*}
The upper bounds for R{\'e}nyi relative entropy of entanglement $E_{R,p}$ is similar. Note that for a separable state $\si^{AB}=\sum_{i}p(i)\si^A_i\ten \si^B_i$,
\[id_A\ten \E_{\mathcal{L}(X)}(\si^{AB})=\sum_{i}p(i)\si^A_i\ten \E_{\mathcal{L}(X)}(\si^B_i)\]
is again a separable state in $\B(H_A)\ten \mathcal{L}(X)\subset \B(H_A\ten H_B)$. Let us denote $S(H_A:\mathcal{L}(X))$ for separable states in $\B(H_A)\ten \mathcal{L}(X)$. Then
\begin{align*}
E_{R,p}(\omega)&=\inf_{\si\in S(A:B)} D_p(\omega|| \si)
\ge  \inf_{\si\in S(A:B)} D_p(\omega||id_A\ten \E_{\mathcal{L}(X)}(\si))
\\ &\ge \inf_{\si\in S(A:\mathcal{L}(X))} D_p(\omega|| \si)
\ge \inf_{\si\in S(A:B)} D_p(\omega||\si)\pl.
\end{align*}
Thus, $E_{R,p}(\omega)
= \inf_{\si\in S(A:\mathcal{L}(X))} D_p(\omega|| \si)$. Again by Theorem \ref{local},
\begin{align*} E_{R,p}(\omega_f)
&= \inf_{\si\in S(A:B)} D_p(\omega_f||\si)
\le\inf_{\si\in S(A:\mathcal{L}(X))} D_p(\omega_f||\si)
\\ &\le \inf_{\si\in S(A:\mathcal{L}(X))} D_p(\omega||\si) +p'\log \norm{f}{\tau ,p}
=E_{R,p}(\omega)  +p'\log \norm{f}{\tau ,p}\pl,
\end{align*}
which completes the proof.
\end{proof}
The next corollary is the comparison property for strong converse rates.
\begin{cor}\label{strong}Let $\N$ be a TRO channel with Stinespring space $X$ and $f$ be a symbol of $\N$. Assume that ${X\cong \oplus_i M_{n_i,m_i}}$, then \begin{enumerate}
\item[i)]$\log ( \sum_{i}n_i)\le C^\dag(\N_f)\le \log ( \sum_{i}n_i)+\tau(f\log f)\pl; $
\item[ii)]$\log (\max_i  n_i)\le Q^\dag(\N_f)\le \log (\max_i  n_i)+\tau(f\log f)$;
\item[iii)]$\log (\max_i  n_i)\le P^\dag(\N_f)\le \log (\max_i  n_i)+\tau(f\log f)$.
\end{enumerate}
\end{cor}
\begin{proof} When $f=1$ and $\N_f=\N$, it corresponds to the formulae given in the Example \ref{cd}. Taking the supremum of all inputs $\rho^{XA'}$ for \eqref{pentropyin}, we have
\[\chi_p(\N_f)\le \chi_p(\N)+p'\log \norm{f}{\tau ,p}\pl.\]
The upper bound of $C^\dag(\N_f)$ follows from regularization based on the upper estimate \eqref{alholevo},
\begin{align*}C^\dag(\N_f)&\le \lim_{k\to \infty} \frac{1}{k}\chi_p(\N_f^{\ten k})
\le  \lim_{k\to \infty} \frac{1}{k}(\chi_p(\N^{\ten k})+ p'\log \norm{f^{\ten k}}{\tau^k ,p})
\\&\le \lim_{k\to \infty}\frac{1}{k}\Big(k\chi_p(\N)+ kp'\log \norm{f}{\tau ,p}\Big)
= \log \big( \sum_{i}{n_i}\big)+ p'\log \norm{f}{\tau ,p}
\end{align*}
where we used the facts that
\[\norm{f^{\ten k}}{\tau^n ,p}=\norm{f}{\tau ,p}^k \pl \text{and}\pl \pl  \chi_p(\N^{\ten k})=k\chi_p(\N)\pl .\]
Then taking the limit $p\to 1^+$ yields
\[C^\dag(\N_f)\le \log \big( \sum_{i}{n_i}\big)+\lim_{p \to 1^+}p'\log \norm{f}{\tau ,p}=\log \big( \sum_{i}{n_i}\big) +\tau(f\log f)\]
The argument for $P^{\dag}$ and $Q^{\dag}$ follow similarly with the upper bounds \eqref{Pp} and  Proposition \ref{Qp}.
\end{proof} 
\section{Examples}
\subsection{Random unitary} Random unitary channels are convex combination of unitary conjugation maps. We observe that the random unitary gives a class of TRO-channels if the unitaries form a projective unitary representation of a group.

Let $G$ be a finite group and $\T$ be the unit complex scalars. We write $1$ as the identity element of $G$. A projective unitary representation of $G$ is a map $u$ from $ G$ into unitary group $U(H)$ of some Hilbert space $H$ such that
\[u(g)u(h)=\si(g,h)u(gh) \pl, \pl g,h\in G\pl,\]
where $\si(g,h)$ is a function $\si: G\times G \to \T$. A projective unitary representation is a representation up to a phase factor, or into the quotient $U(H)/\T$.
Because the laws of group multiplication, the function $\si$ satisfies the following conditions
\begin{enumerate}
\item[i)] $\si(g,1)=\si(1,g)=1\pl ;$
\item[ii)] $\si(g,g')\si(gg',g'')=\si(g,g'g'')\si(g',g'')\pl ;$
\end{enumerate}
for all $g,g',g''\in G$. Suppose $|G|=n$ and $dim H=m$ are finite. We define a $m$-dimensional channel $\N: \B(H) \to \B(H)$ as follows
\[\pll \N(\rho)=\frac{1}{n}\sum_g u(g)\rho u(g)^* \pl.\]
Its Stinespring isometry is given by
\[V: H \to H\ten l_2(G) \pll, \pll V\ket{h}= \frac{1}{n}\sum_{g} u(g)\ket{h}\ten \ket{g} \pl ,\]
where $l_2(G)$ is the Hilbert space spanned by the canonical orthogonal basis $\{\ket{g}\pl |\pl g\in G\}$. The Stinespring space $X$, as a subspace of operators $\B(l_2(G),H)$, is
\[X=ran(V)=\{\pl \sum_{g} u(g)\ket{h}\ten \bra{g}\pl | \pl \ket{h}\in H  \pl \} \pl.\]
We claims that $X$ is a TRO space. Let $\ket{h_1},\ket{h_2},\ket{h_3}$ be vectors in $H$ and $h,h_1,h_2$ be the corresponding operators in $X\subset \B(l_2(G),H)$
\begin{align*}
h_1h_2^*h_3=&\sum_{g,g',g''}
u(g)\ket{h_1}\bra{g}g'\ran \bra{h_2}u^*(g')u(g'')\ket{h_3}\bra{g''}\\
=&\sum_{g,g''}
u(g)\ket{h_1} \bra{h_2}u^*(g)u(g'')\ket{h_3}\bra{g''}\\
=&\sum_{g''}n\N(\ket{h_1} \bra{h_2})u(g'')\ket{h_3} \bra{g''}\\
=&\sum_{g''} u(g'')\Big(nN(\ket{h_1} \bra{h_2})\ket{h_3} \Big)\bra{g''}
\end{align*}
In the last step, we use the fact $\N$ is the conditional expectation onto the commutant $u(G)'\subset \B(H)$. Indeed, for $g_0\in G$,
\begin{align*}\N(\rho) u(g_0)&=\frac1n\sum_{g}u(g)\rho u(g)^*u(g_0)
=\frac1n\sum_{g}u(g)\rho u(g)^*u(g_0^{-1})^*\si(g_0,g_0^{-1})\\&= \frac{1}{n}\sum_{g}u(g)\rho u(g_0^{-1}g)^*\overline{\si(g_0^{-1},g)}\si(g_0,g_0^{-1})
\\&=\frac{1}{n}\sum_{g}u(g_0)u(g_0^{-1}g)\rho u(h^{-1}g)^*\overline{\si(g_0,g_0^{-1})\si(g_0^{-1},g)}\si(g_0,g_0^{-1})
\\&= \frac{1}{n}\sum_{g}u(g_0)u(g_0^{-1}g)\rho u(g_0^{-1}g)^* =u(g_0)\N(\rho)\pl.
\end{align*}
Here we use the conditions of the phase factor $\si$,
\begin{align*}\overline{\si(g_0,g_0^{-1}g)\si(g_0^{-1},g)}\si(g_0,g_0^{-1})&=\overline{\si(1,g)\si(g_0,g_0^{-1})\si(g_0,g_0^{-1})}
\\&=\overline{\si(1,g)\si(g_0,g_0^{-1})}\si(g_0,g_0^{-1})=1
\end{align*}
Thus we verify that $X$ is a tenary ring of operators in $\B(l_2(G),H)$. The left algebra
$\mathcal{L}(X)=ran(\N)$ is exactly the commutant $u(G)'$. For the right $C^*$-algebra $\mathcal{R}(X)$,
\begin{align*}h_1^*h_2&=\sum_{g,g_0}\ket{gg_0^{-1}}\bra{g} \bra{h_1}u(gg_0^{-1})^*u(g)\ket{h_2}
\\ &=\sum_{g,g_0}\ket{gg_0^{-1}}\bra{g} \bra{h_1} \si(g,h^{-1})u(g_0)^*\ket{h_2}
\\ &=\sum_{g_0}\bra{h_1} u(g_0)^*\ket{h_2}(\sum_g\si(g,g_0^{-1})\ket{gg_0^{-1}}\bra{g} ) \pl.
\end{align*}
This gives an element in the $\si$-twisted right regular representation $\pi_\si: G \to B(l_2(G))$,
\[\rho_\si(g_0)\ket{g}=\si(g,g_0^{-1})\ket{gg_0^{-1}} \pl.\]
Thus $\mathcal{R}(X^\N)\subset \rho_\si(G)$ as a subalgebra. The diagonal matrices $l_\infty(G)$
 is an algebra independent of $\rho_\si(G)$. Indeed, for $f=\sum_g f(g)\ket{g}\bra{g}$ be a diagonal matrix in $l_\infty(G)$ and $x=\sum_g \al(g)\rho_\si(g)$ be a element in $\pi_\si(G)$,
\[tr(fx)=\frac{1}{n}tr(f)\al(1) =tr(f)\tau(x)\pl.\]
Given a normalized density $f\in l_\infty(G)$ $(\sum f(g)=|G|, f\ge 0)$, the channel $\N_f$ is
\[\N_f(\rho)=\frac{1}{|G|}\sum_g f(g)u(g)\rho u(g)^*\pl.\]
It is clear from above that all our estimates apply here. Assume that $u(G)'=\oplus_k M_{n_i}\ten 1_{m_i}$. $m_i$'s are the dimensions of the irreducible decomposition of $u$ and $n_i$'s are corresponding multiplicities. Then
\[C(\N)=\log \big(\sum_i n_i\big) \pl , \pl Q(\N)=P(\N)=\log \big(\max_i n_i\big) \pl,\]
and for all normalized densities $f\in l_\infty(G)$,
\begin{enumerate}
\item[i)]$\log \big(\sum_i n_i \big)\le C(\N_f)\le C^{\dag}(\N_f) \le \log \big(\sum_i n_i\big)+\tau(f\log f)\pl;$
\item[ii)]$\log (\max_i n_i)\le Q(\N_f)\le P(\N) \le P^{\dag}(\N_f) \le \log \big(\max_i n_i\big)+\tau(f\log f)\pl.$
\end{enumerate}
When the group $G$ is noncommutative, then the dimensions $m_i$ of irreducible representations may greater then $1$. In this situation, the random unitary channel $\N_f$ are in general not degradable because they depolarize matrix blocks of nontrivial size.
\subsection{Generalized dephasing channels}
Generalized dephasing channels are also called Schur multipliers in the literature (see e.g. \cite{paulsen}). They are special cases of Hadamard channels which are known to be degradable, hence the quantum capacity and private capacity does not require regularization, and $Q^{(1)}=Q=P$. Our estimates here recovers the quantum capacity formula in \cite{Neufang} in a different way. Both approaches are based on the unpublished joint work \cite{JNR2}. Our approach provides a new proof of $Q=Q^{(p)}$ for these particular Schur multipliers. This is already known \cite{WD} thanks to the fact that Hadamard channels are strongly additive for $Q^{(1)}$.

The Schur multiplication (or Hadamard product) of matrices is given by
\begin{align*}(a_{ij})*(b_{ij})=(a_{ij}\cdot b_{ij}).\end{align*}
It is a well-known fact (see \cite{paulsen}) that the multiplier map for a given matrix $a=(a_{ij})$,
\begin{align*}M_a (b)=a* b\ \ \ \text{for}\  b=(b_{ij})\in M_n\pl ,\end{align*}
is completely positive if and only if $a$ is positive. Clearly, $M_a$ is trace preserving if and only if $a_{ii}=1$ for $1\le i\le n$.

Let $G$ be a finite group with order $|G|=n$. A function $f:G\to \CC$ is positive definite if for any finite sequence $g_1,g_2,\cdots,g_k\in G$, the matrix $(f(g_j^{-1}g_i))_{i,j=1}^{k}$
is positive. Consider the Schur multiplier
\[\M_f: \B(l_2(G))\to \B(l_2(G))\pl,\pl \M_f(\ket{g}\bra{g'})=f(g'^{-1}g)\ket{g}\bra{g'} \pl.\]
$\M_f$ is completely positive if $f$ is positive definite and $f(1)=1$. In particular, the function
$\delta(g)=\begin{cases}
                                     1, & \mbox{if } g=1 \\
                                     0, & \mbox{otherwise}.
                                   \end{cases}$ gives the completely dephasing channel
\[ \M_\delta(\sum_{g,g'} a_{g,g'}\ket{g}\bra{g'})=\sum_{g}a_{gg}\ket{g}\bra{g}\pl.\]
This is a TRO channel as its Stinespring dilation is given by
\[V:l_2(G)\to l_2(G)\ten l_2(G)\pl, \pl V\ket{g}=\ket{g}\ten \ket{g}\]
The Stinespring space, via the identification
$\ket{g}\ten \ket{g'}\longleftrightarrow \ket{g}\bra{g'}$, is the diagonal matrices
$l_\infty(G)\subset \B(l_2(G))$. $\M_f$ can be written as a modified channel of symbol $f=\sum_{g,g'}f(g'^{-1}g)\ket{g}\bra{g'}$ as follows,
\[\M_f(\ket{g}\bra{g'})=\ket{g}\bra{g}f\ket{g'}\bra{g'}\pl.\]
Such an operator $f\in \B(l_2(G))$ belongs to the right regular representation and is strongly independent to $l_\infty(G)$. Note that $\M_\delta$ is a channel with commutative range hence has $C(\M_\delta)=\log |G|, Q(\M_\delta)=P(\M_\delta)=0$. Therefore, for any Schur multiplier given by positive definite functions, Theorem \ref{comparison} gives
\[ Q(\M_f)\le P(\M_f)\le \tau(f\log f)=\log |G|-H(\frac{1}{|G|}f)\pl.\]
Recall that the negative $cb$-entropy $-S_{cb}$ is a lower bound for $Q^{(1)}$ for unital channels $\N$. Then Theorem \ref{negativecb} gives the lower bound via $-S_{cb}(\M_f)=\tau(f\log f)$ hence we have the formula
 \[Q(\M_f)= P(\M_f)=\log |G|-H(\frac{1}{|G|}f) \pl,\]
which recovers the formula from \cite{Neufang} in a different way.
\begin{exam}\label{dephasing}{\rm
The qubit example is the dephasing channel. Let $0\le q\le 1$ be the dephasing parameter, we have
\begin{align*}
\Phi_q \large(\big[ \begin{array}{cc}
    a&b\\
    c&d
  \end{array}\big]\large)=\big[\begin{array}{cc}
    a&qb\\
   q c&d
  \end{array}\big]\pl .
\end{align*}
This corresponds to $G=\mathbb{Z}_2$ for $f=\big[\begin{array}{cc}
    1 &q\\
   q &1
  \end{array}\big]$ in our setting. The formula for the quantum capacity is $Q(\Phi_q)=\log 2-H(\fs{1+q}{2})=\tau(f\log f)$.}
\end{exam}
When the dimension $d>2$, not every generalized dephasing channel can be expressed via positive definite functions.

\subsection{Small dimensional example} We provide a concrete example in small dimensions which are nondegradable channels and our upper bound are tight. Let $|\al|\le 1$ be a real number. Define the channel $\Phi_\al :M_4\to M_3$ as follows,
\begin{align*}
\Phi_\al(\left[
\begin{array}{cccc}
    a_{11}& a_{12}&a_{13}&a_{14}\\
   a_{21}& a_{22}&a_{23}&a_{24}\\
   a_{31}& a_{32}&a_{33}&a_{34}\\
   a_{41}& a_{42}&a_{43}&a_{44}
  \end{array}\right] )=\left[
\begin{array}{ccc}
    a_{11}+a_{22}& \al a_{13}&\al a_{24}\\
   \al a_{31}&a_{33}&0\\
   \al a_{42}&0&a_{44}
  \end{array}\right]
\end{align*}
This channel is non-degradable since it traces out the first $2\times 2$ block. We claim that
\[Q^{(1)}(\Phi_\al)=Q^{(p)}(\Phi_\al)=Q^\dag(\Phi_\al)=P^{(1)}(\Phi_\al)=P^{(p)}(\Phi_\al)=P^{\dag}(\Phi_\al)=1-h(\frac{1+\al}{2})\pl,\]
where $h(\la)=-\la\log\la -(1-\la)\log (1-\la)$ is the binary entropy function. Let us first consider the diagonal part of the channels. That is when $\al=0$,
\begin{align*}
\Phi_0\kla \left[
\begin{array}{cccc}
    a_{11}& a_{12}&a_{13}&a_{14}\\
   a_{21}& a_{22}&a_{23}&a_{24}\\
   a_{31}& a_{32}&a_{33}&a_{34}\\
   a_{41}& a_{42}&a_{43}&a_{44}
  \end{array}\right] \mer=\left[
\begin{array}{ccc}
    a_{11}+a_{22}& 0&0\\
   0&a_{33}&0\\
   0&0&a_{44}
  \end{array}\right]
\end{align*}
It is an orthogonal sum of partial trace maps hence the Stinespring space corresponds to a TRO. Let $\{e_i\}$ be the standard (computational) basis. The Stinespring isometry $V_0$ of $\Phi_0$ is given by
\[V(\sum_{i}h_ie_i)= h_1e_1\ten e_1+h_2e_1\ten e_2+h_3e_2\ten e_3+h_4e_3\ten e_4 \in \mathbb{C}^3\ten \mathbb{C}^4 \pl.\]
The corresponding operators are $3\times 4$ matrices,
$h=  \left[\begin{array}{cccc}
  h_1& h_2&0&0\\
   0& 0&h_3&0\\
   0& 0&0&h_4
  \end{array}\right] $.
Then the Stinespring space $X=M_{1,2}\oplus\mathbb{C}\oplus\mathbb{C}$ as a TRO.
The left and right algebra are given by
\[\mathcal{L}(X)=\mathbb{C}\oplus\mathbb{C}\oplus\mathbb{C}\pl ,\pl \mathcal{R}(X)=M_2\oplus\mathbb{C}\oplus\mathbb{C}\pl.\]
Let $S=\left[\begin{array}{cccc}
    0& 0&1&0\\
 0& 0&0&1\\
   1& 0&0&0\\
  0& 1&0&0
  \end{array}\right] $.
One verifies that the only nontrivial $*$-subalgebra independent of $\mathcal{R}(X)$ in $M_4$ is
\[N=\{\pl\beta I+\al S \pl|\pl \al, \beta\in \mathbb{C}\}\pl.\]
The normalized densities in $N$ given by the one-parameter class $\{I+\al S|-1\le \al \le 1\}$. Denote the symbol $f_\al=I+\al S$. Note that $f_0$ is the identity $1_E$. $\Phi_\al$ is a modified TRO channel with symbol $f_\al$,
\[\Phi_\al(\ket{h}\bra{h})= hfh^*=\left[
\begin{array}{cccc}
    h_1& h_2&0&0\\
   0& 0&h_3&0\\
   0& 0&0&h_4

  \end{array}\right]
\left[  \begin{array}{cccc}
    1& 0&\al&0\\
 0& 1&0&\al\\
   \al& 0&1&0\\
  0& \al&0&1
  \end{array}\right]
  \left[ \begin{array}{ccc}
    \overline{h_1}& 0&0\\
    \overline{h_2}& 0&0\\
   0&  \overline{h_3}&0\\
 0& 0& \overline{h_4}
  \end{array}\right] \pl.\]
Via a change of basis, one can identify $f=I_2 \ten  \left[ \begin{array}{cc}
    1+\al& 0\\
   0& 1-\al\\
\end{array}\right]$. Thus for the entropy term we have $\tau(f\log f)=1-h(\frac{1+\al}{2})$. Since $\Phi_0$'s outputs are all diagonal matrices, then
  \[Q^{(p)}(\Phi_0)=P^{(p)}(\Phi_0)=Q^{\dag}(\Phi_0)=P^{\dag}(\Phi_0)=0 \pl.\]
By our comparison estimates (Correllary \ref{comparison} and \ref{strong}), we obtain that $1-h(\frac{1+\al}{2})$ is an upper bound for $Q^{(p)}(\Phi_\al),P^{(p)}(\Phi_\al),Q^{\dag}(\Phi_\al)$ and $P^{\dag}(\Phi_\al)$. On the other hand, $1-h(\frac{1+\al}{2})$ is the quantum capacity of a qubit dephasing channel with parameter $\al$,
\begin{align*}
\Psi_\al\kla \left[
\begin{array}{cc}
    a_{11}& a_{12}\\
   a_{21}& a_{22}\\
  \end{array}\right] \mer=\left[
\begin{array}{cc}
    a_{11}& \al a_{12}\\
    \al a_{21}&a_{22}\\
  \end{array}\right]\pl,
\end{align*}
which can be implemented in $\Phi_\al$ by using the block input
\[\left[
\begin{array}{cccc}
    a_{11}& 0&a_{13}&0\\
  0& 0&0&0\\
   a_{31}& 0&a_{33}&0\\
   0& 0&0&0
  \end{array}\right] \pll \pl \text{or} \pl \pll
  \left[
\begin{array}{cccc}
   0&  0&0&0\\
  0&  a_{22}&0& a_{24}\\
   0& 0&0&0\\
   0&  a_{42}&0& a_{44}
  \end{array}\right] \pl.
\]
By the fact $Q\le P,Q^{(p)},P^{(p)},Q^{\dag},P^{\dag}$, the upper bound $1-h(\frac{1+\al}{2})$ is achievable.


\emph {Acknowledgements}--- We thank Mark M. Wilde for helpful discussions and comments. 

\section{Appendix: Complex interpolation and Noncommutative $L_p$ spaces}
In this Appendix, we briefly review the complex interpolation theory that is used in the proof of Theorem \ref{pcom}. The readers are referred to \cite{BL} for interpolation theory and \cite{pisier93} for vector-valued noncommutative $L_p$ spaces.

Two Banach spaces $X_0$ and $X_1$ are compatible if there exists a Hausdorff topological vector space $X$ such that $X_0, X_1\subset X$ as subspaces. The sum space $X_0+X_1$ is a Banach space
\[X_0+X_1:\lel \{x\in X\pl |\pl x=x_0+x_1\pl \text {for some}\pl x_0\in X_0, x_1\in X_1\}\pl,\]
equipped with the norm
\[\norm{x}{X_0+X_1}=\inf_{x=x_0+x_1} (\norm{x_0}{X_0}+\norm{x_1}{X_1})\pl.\]
Let $S=\{z|0\le Re (z)\le 1\}$ be the vertical strip of unit width on the complex plane, and let $S_0=\{z|0< Re (z)< 1\}$ be its open interior. We denote by $\F(X_0, X_1)$  the space of all functions $f:S\to X_0+X_1$, which are bounded and continuous on $S$ and analytic on $S_0$, and moreover
\[\{f(it)\pl|\pl t\in \mathbb{R}\}\subset X_0\pl ,\pl \{f(1+it)\pl|\pl t\in \mathbb{R}\}\subset X_1\pl.\]
$\F(X_0, X_1)$ is again a Banach space with the norm
\[\norm{f}{\F}=\max\{\pl \sup_{t\in \mathbb{R}} \norm{f(it)}{X_0}\pl,\pl \sup_{t\in \mathbb{R}}\norm{f(1+it)}{X_1}\}\pl. \]
The complex interpolation space $(X_0,X_1)_\theta$, for $0<\theta<1$, is the quotient space of $\F(X_0,X_1)$ given as follows,
\[(X_0, X_1)_\theta=\{\pl x\in X_0+X_1\pl| \pl x=f(\theta)\pl, \pl f\in \F(X_0, X_1)\pl\} \pl.\]
The quotient norm is defined as
\[\norm{x}{\theta}=\inf \{\pl\norm{f}{\F}\pl| \pl f(\theta)=x \pl\}\pl .\]
For example, the Schatten-$p$ class is the interpolation space of bound operator and trace class
\[S_p(H)=(B(H), S_1(H))_{\frac{1}{p}}\pl.\]
 This generalizes to vector-valued noncommutative $L_p$-space $S_p(A,S_q(B))$ (see \cite{pvp}). In particular, for any $1\le p, q\le \infty$ one has the relations
\begin{align*}&S_p(A, S_q(B))=[S_\infty(A, S_q(B)),S_1(A, S_q(B))]_{\frac1p} \pl ,\\ \pl &S_p(A, S_q(B))=[S_\infty(A, S_\infty(B)),S_1(A, S_1(B))]_{\frac1q} .\end{align*}
The following Stein's interpolation theorem (cf. \cite{BL}) is a key tool in our analysis. \begin{theorem}\label{stein}
Let $(X_0,X_1)$ and $(Y_0,Y_1)$ be two compatible couples of Banach spaces. Let $\{T_z| z\in S\}\subset \B(X_0+X_1, Y_0+Y_1)$ be a bounded analytic family of maps such that
\[\{T_{it}|\pl t\in \mathbb{R}\}\subset \B(X_0,Y_0)\pl ,\pl \{T_{1+it}|\pl t\in \mathbb{R}\}\subset  \B(X_1,Y_1)\pl.\]
Suppose $\Lambda_0=\sup_t{\norm{T_{it}}{\B(X_0,Y_0)}}$ and  $\Lambda_1=\sup_t{\norm{T_{1+it}}{\B(X_1,Y_1)}}$ are both finite, then for $0< \theta < 1$, $T_\theta$ is a bounded linear map from $(X_0,X_1)_\theta$ to $(Y_0,Y_1)_\theta$ and
\[\norm{T_\theta}{\B((X_0,X_1)_\theta ,(Y_0,Y_1)_\theta)}\le \Lambda_0^{1-\theta}\Lambda_1^{\theta}\pl .\]
\end{theorem}
\noindent In particular, when $T$ is a constant map, the above theorem implies
\begin{align}
\norm{T}{\B((X_0,X_1)_\theta ,(Y_0,Y_1)_\theta)} \le \norm{T}{\B(X_0 ,Y_0)}^{1-\theta}\norm{T}{\B(X_1 ,Y_1)}^{\theta} \pl . \label{interpolation1}
\end{align}

Let $X\subset \B(H,K)$ be a TRO. Denote $X_p$ the closure of intersection $X\cap S_p(H,K)$ in the Schatten $p$-class. TROs $X$ and their corresponding subspaces $X_p$ in $S_p(H,K)$ are completely $1$-complemented for all $ p\in [1,\infty]$ (see \cite{EOR2001,NO2002}). That is, there exists a projection map $\mathcal{P}$ from $\B(H,K)$ (resp. $S_p(H,K)$) onto $X$ (resp. $X_p$) such that $id_{n}\ten \mathcal{P}$ is contractive for every $n$. A direct consequence is that $X_p$ are interpolation spaces of $X_1$ and $X=X_\infty$,
\[X_p=(X_\infty, X_1)_{\frac{1}{p}}\pl.\]
In the proof of Theorem \ref{pcom}, we used the following simple application of Kosaki-type interpolation \cite{Kosaki1984}.
\begin{theorem}\label{interpolation}
Let $X$ be a TRO. For a positive operator $\si\in \mathcal{L}(X)$ and $1\le p\le \infty$, define $X_{p,\si}$ as the space $X$ equipped with the following norms,
\[\norm{x}{p,\si}:\lel\norm{\si^{\frac{1}{p}}x}{p}\pl.\]
Then
\[[X_{\infty}, X_{1,\si}]_{\frac{1}{p}}=X_{p,\si}\pl.\]
\end{theorem}
\begin{proof}Let us first assume that $\si$ is invertible. For $x\in X$ such that $\norm{\si^{\frac{1}{p}}x}{p}=1$, we consider the polar decomposition $\si^{\frac{1}{p}}x=v|\si^{\frac{1}{p}}x|:= vy$, where $v \in X$ is a partial isometry and $y\in \mathcal{L}(X)$. Then we define the analytic function $x$ from the strip ${S=\{z|\pl 0\le Re (z)\le 1\}}$ to $X$ as follows,
\[x(z)=\si^{-z}vy^{pz}\pl \pl , \pl\pl x(1/p)=\si^{-\frac{1}{p}} vy=x \pl.\]
Note that \[\norm{x(it)}{\infty}=\norm{\si^{-it}vy^{itp}}{\infty}\le 1 \pl, \pl\norm{x(1+it)}{1,\si}=\norm{\si \si^{-1-it}vy^{p(1+it)}}{1}=\norm{vy^{p}}{1}\le\norm{vy}{p}^p\le1. \]
Therefore $\norm{x}{[X_{\infty}, X_{1,\si}]_{\frac{1}{p}}}\le \norm{\si^{\frac{1}{p}}x}{p}$. On the other hand, suppose that we have an analytic function $x: S \to X$ such that
\[\sup_{t}\{\pl\norm{x(it)}{\infty}, \norm{\si x(1+it)}{1}\} \le 1 \pl .\] Recall that $1/p'+1/p=1$. For any $\norm{a}{S_{p'}(K,H)}\le 1$, we claim
\[tr(\si^{1/p}xa)\le 1\pl.\]
Indeed, consider the analytic function
$h(z)=tr(\si^{z}x(z)a(z))$ where $a(z)=w|a|^{p'(1-z)}$.
On the boundary of the strip $S$, \[|h(it)|\le\norm{x(it)}{\infty}\norm{a(it)}{1}\le 1 \pl ,\pl\pl |h(1+it)|\le\norm{\si x(1+it)}{1}\norm{a(1+it)}{\infty}\le1 .\]
By the maximum principle, we obtain that $|h(1/p)|=|tr(\si^{\frac{1}{p}}xy)|\le  1$, which proves the claim. For noninvertible $\si$, one can repeat the argument for $\tilde{\si}=\si+\delta 1$ with $\delta>0$ and let $\delta$ go to $0$.
\end{proof}
\begin{rem}{\rm The above interpolation relation can be generalized to two-sided densities. Let $0\le \theta\le 1$. Given $\si\in \mathcal{L}(X), \rho\in \mathcal{R}(X)$, one can define $X_{p,\theta}$ as the corresponding space equipped with the norm,
\[\norm{x}{p,\theta,\si,\rho}\lel\norm{\si^{\frac{\theta}{p}}x\rho^{\frac{1-\theta}{p}}}{p}\pl.\]
These $L_p$-spaces also interpolate \cite{JX},
\[[X_{\infty}, X_{1,\theta,\si,\rho}]_{\frac{1}{p}}=X_{p,\theta,\si,\rho}\pl.\]}
\end{rem}

\end{document}